\documentclass[12pt]{article}

\usepackage{graphicx}
\usepackage{amssymb}

\newcommand{\blackslug}{\hbox{\hskip 1pt
        \vrule width 4pt height 8pt depth 1.5pt\hskip 1pt}}
\newcommand{\myQED}{\hfill \blackslug}
\newcommand{\la}{\langle}
\newcommand{\ra}{\rangle}
\newcommand{\mss}{\hspace*{0.05in}}
\newcommand{\mbs}{\hspace*{0.2in}}

\newtheorem{observation}{Observation}

\newenvironment{proof}
    {\pagebreak[1]{\narrower\noindent {\bf Proof:\nopagebreak}}}%
    {\myQED}

    {\myQED}

\newtheorem{lemma}{Lemma}

\newtheorem{definition}{Definition}

\begin{document}

\begin{center}
{\large \bf
Viable Algorithmic Options for \\
Creating and Adapting Emergent Software Systems}

\vspace*{0.2in}

Todd Wareham \\
Department of Computer Science \\
Memorial University of Newfoundland \\
St.\ John's, NL Canada \\
(Email: {\tt harold@mun.ca}) \\

\vspace*{0.1in}

Ronald de Haan \\
Institute for Logic, Language, and Computation \\
Universiteit van Amsterdam \\
Amsterdam, The Netherlands \\
(Email: {\tt me@ronalddehaan.eu}) \\
\vspace*{0.1in}

\today
\end{center}

\begin{quote}
{\bf Abstract}:
Given the complexity of modern software systems,
it is of great importance that such systems
be able to autonomously modify themselves, i.e., self-adapt, 
with minimal human supervision. It is critical that this adaptation both results
in reliable systems and
scales reasonably in required memory and runtime to non-trivial systems.
In this paper, we apply computational complexity analysis to evaluate 
algorithmic options for the reliable creation and adaptation of emergent 
software systems
relative to several popular types of exact and approximate
efficient solvability. We show that neither problem is 
solvable for all inputs when no restrictions are placed
on software system structure.
This intractability continues to hold 
relative to all examined types of efficient exact and approximate solvability
when software systems are restricted to run (and hence can be verified
against system requirements) in polynomial time.
Moreover, both of our problems when so restricted remain intractable under
a variety of additional restrictions on software system structure,
both individually and in many combinations. That being said, we also give
sets of additional restrictions that do yield tractability for both
problems, as well as circumstantial evidence that emergent software
system adaptation is computationally easier than emergent software
system creation.
\end{quote}

\section{Introduction}

\label{SectIntro}

Given the complexity of modern software systems and the increasing need to 
modify such systems to handle both unplanned changes in system requirements and 
varying operating environments, such systems must
be able to autonomously modify themselves, i.e., self-adapt, 
with minimal human supervision \cite{CdL+09,KR+15,ST09}. Over the last 25 years,
a great deal of research has been done on self-adaptive software systems, and 
a number of such systems based on various types of adaptation controllers
(e.g., MAPE-K feedback loop \cite{IW15}, model-predictive \cite{AP+18}, 
general control-theoretic \cite{FM+17}) have been created (see also 
\cite{CdL+09,dLG+17,dLG+13}). Certifying that adaptation does not
cause the resulting systems to violate functional and/or non-functional system 
requirements, i.e., validation and verification (V\&V), is 
a major research challenge \cite[page 21]{dLG+13}. Such V\&V must be done
quickly enough that the proposed adaptations are still valid (to handle
rapidly-changing operating environments). This is further complicated by the
requirement that enough computational effort be put into the 
adaptation space search process to ensure that the adapted system optimizes 
system performance as much as possible (to satisfy system user expectations).

A valuable complement to algorithm development work 
would be to establish what the general 
algorithmic options are for software system creation and adaptation algorithms 
that are guaranteed to satisfy specified functional and non-functional
system requirements while optimizing adapted system performance. This can  be
done using the tools and techniques of computational complexity analysis
\cite{DF13,GJ79,Gol08}. The results of such analyses can be used not only 
to establish those situations in which known algorithms are the best possible 
but also to guide the development of new algorithms (by highlighting relative 
to which types of efficient solvability such algorithms can and cannot exist).

A good test case for the utility of such complexity analyses would be
a 
type of self-adaptive systems called emergent software
systems \cite{FP17,PG+16}. Such systems initially self-assemble from a provided library 
of software components to satisfy basic system requirements, and then 
continuously self-adapt with the assistance of a learning algorithm that 
uses collected system environment events
and performance metrics to optimize a reward function such as
system runtime or memory usage. An emergent web server based on a
library of 30 components has been constructed under this paradigm and seems
to function well \cite{PG+16}. However, as noted on \cite[page 14]{FP17},

\begin{quote}
\ldots [the] learning algorithm, which is based on an exhaustive exploration 
phase, is not designed to scale up to large systems with thousands of 
compositions, but rather serves as a proof of concept and useful baseline 
against which to compare more sophisticated algorithms. This is an active 
research area for which we continue to develop more efficient and scalable 
solutions \ldots
\end{quote}
 
\noindent
Ongoing work has focused on exploiting such strategies as
making components encode single responsibilities and dividing large systems into
distributed collections of smaller (and hopefully more manageable) subsystems
\cite{FP+18,PF19}. It is precisely at this current stage of new algorithm
development that computational complexity analyses might be most useful. 

\subsection{Previous Work}

\label{SectPrevWork}

Computational complexity analyses have been done previously for 
component-based software
system creation by component selection \cite{PO99,PWM03} and component selection
with adaptation \cite{BBR05}, with \cite{BBR05,PWM03} having the additional 
requirement that the number of components in the resulting software system be 
minimized. Given the intractability of all of these problems, subsequent work 
has focused on efficient approximation algorithms for component selection. 
Though it has been shown that efficient algorithms that produce software systems
whose number of components is within a constant multiplicative factor of optimal
are not possible in general \cite{NH08}, efficient approximation algorithms 
are known for a number of special cases \cite{FBR04,HM+07,NH08}. All of these 
analyses assume that any component can be composed with any other component, 
in the sense that function definitions not given in a component $c$ can be obtained
by composition of $c$ with any other components that have the required function definitions,
i.e., component composition is not regulated using component interfaces.
Moreover, none of the formalizations used include specification of the internal
structure of software systems, system requirements, and components that are 
detailed enough to allow investigation of restrictions on these aspects that 
could make component selection or component selection with adaptation tractable.

Only one analysis to date has incorporated a component model which allows 
investigation of the tractability effects of restrictions on system requirement,
component internal, and software system structure, namely that in
in \cite{WS15} (subsequently reprinted as \cite{WS16}; see also \cite{WS22} for the full version with proofs of results). The focus in this work
was on exact polynomial-time solvability and fixed-parameter tractability of 
component-based software system creation and adaptation, where system adaptation
is in response to changes in the functional system requirements. The
authors showed that both
of these problems are intractable in general and remain so under a variety of
restrictions on system requirement, component, and software system structure. 
This was done relative to radically restricted components (single {\tt 
if-then-else} blocks) and software systems (two-level reactive systems). While
components in this model cannot arbitrarily compose in the sense described 
above, component composition was implicitly regulated by the the combination of 
two-level system structure and a parameter restricting the number of types of components in a
system, and there was no notion of a component interface, let alone interfaces providing 
state variables or functions with input
parameters and/or returned values.

\subsection{Summary of Results}

In this paper, we present the first computational complexity analyses of the
problems of emergent software system creation and adaptation. These problems can 
be stated informally as follows (and are described in more detail in Section \ref{SectForm}):

\begin{description}
\item[{\sc Emergent Software Creation} (ESCreate):] Derive (possibly by using a system
        environment function $Env()$) an emergent 
        software system $S$ relative to given libraries $L_{int}$ and $L_{comp}$ of 
        software interfaces and components
        that satisfies a given set $R$ of software requirements and has the best possible
        value for a specified reward function $Rew()$.
\item[{\sc Emergent Software Adaptation} (ESAdapt):] Given an emergent software system
        $S$ relative to given libraries $L_{int}$ and $L_{comp}$ of software interfaces and
        components that satisfies a given set of software system requirements $R$, derive 
        (possibly by using a system environment function $Env()$) an
        emergent software system $S'$ relative to $L_{int}$ and $l_{comp}$
        that satisfies $R$ and has the best possible
        value for a specified reward function $Rew()$.
\end{description}

\noindent
We consider the following types of efficient solvability (described in more detail in
Section \ref{SectResDesSolv}):

\begin{enumerate}
\item{
Polynomial-time exact solvability, such that a polynomial-time algorithm 
produces the correct output for a given input either (a) all the time 
\cite{GJ79} or (b) when such an output is known to exist (promise
solvability).
}
\item{
Polynomial-time approximate solvability, such that a polynomial-time algorithm
produces the correct output for a given input either (a) in all but a small 
number of cases \cite{HW12} or (b) with a high probability \cite{MR10}, or, in 
the case of a problem that requires an output that optimizes some cost measure,
(c) produces an output for a given input whose cost is within some arbitrarily 
small fraction of optimal \cite{AC+99}.
}
\item{
Effectively polynomial-time exact restricted solvability (e.g., 
fixed-parameter (fp-)\allowbreak tractability \cite{DF13}), such that an
algorithm produces the correct output for a given input in what is
effectively polynomial time when certain aspects of that input are of
restricted value, e.g., the number of components in the given library or
any valid assembled software system is small.
}
\end{enumerate}

\noindent
Relative to these types of solvability and various conjectures that are widely
believed to be true within the Computer Science community, e.g. $P \neq NP$
\cite{For09,GJ79}, we prove the following results (Section \ref{SectRes}):

\begin{itemize}
\item Neither ESCreate nor ESAdapt is solvable by an algorithm (regardless of
       runtime or memory usage) that gives the correct output for every input,
       i.e., both of these problems are unsolvable in the same sense as 
       Turing's classic Halting Problem \cite{Tur36}
       (Sections \ref{SectResESCUnsolv} and \ref{SectResESAUnsolv}, respectively).
\item Neither ESCreate nor ESAdapt is solvable by a polynomial-time algorithm
       in the sense (1a) above and ESCreate is not solvable by a polynomial-time
       algorithm in sense (1b) above, even when software systems are
       restricted to run in and hence can be verified against requirements in
       polynomial time (Sections \ref{SectResESC_PTExactSolv} and 
       \ref{SectResESA_PTExactSolv}, respectively).
\item Neither ESCreate nor ESAdapt is solvable by a polynomial-time algorithm
       in the senses (2a), (2b), or (2c) above, even when software systems are
       restricted to run in and hence can be verified against requirements in
       polynomial time (Sections \ref{SectResESC_PTApproxSolv} and 
       \ref{SectResESA_PTApproxSolv}, respectively).
\item{ Neither ESCreate nor ESAdapt is fixed-parameter tractable in sense (3)
        above relative to restrictions of the values of the following aspects
        of the input:

\begin{itemize}
\item The number of software component interfaces in $L_{int}$.
\item The number of software components in $L_{comp}$.
\item The maximum number of components implementing an interface.
\item The maximum number of interfaces provided by a component.
\item The maximum number of interfaces required by a component.
\item The maximum number of components in a software system.
\end{itemize}

\noindent
This fixed-parameter intractability holds both when software systems are restricted to run 
in and hence can be verified against requirements in polynomial time and relative to
many combinations of the aspects listed above, often when aspects are restricted to
small constant values. That being said, there are several combinations of 
aspect-restrictions that do yield fixed-parameter tractability
(Sections \ref{SectResESC_FPSolv} and \ref{SectResESA_FPSolv}, respectively).
}
\end{itemize}

\noindent
All of the results above hold for ESCreate (except those related to cost-
\linebreak inapproximability)
relative to any choices of $Env()$ and $Rew()$ and for ESAdapt (and ESCreate for 
cost-inapproximability) relative to any choice of $Env()$ and one of two specified
reward functions, namely $Rew_{\#comp}(S)$ (the number of components in a software
system) and $Rew_{CodeB}(S)$ (the total size of the interface and component code
is a software system). This, in combination with the unresolved fixed-parameter status
of certain aspect-combinations at this time, suggests that emergent software adaptation may 
in general be easier to do than emergent software creation (Section \ref{SectDiscImpRW}).

The list above may, at first glance, be read as 
saying that emergent software system creation and adaptation are not possible 
under {\em any} circumstances. However, this bleak interpretation is very much contrary 
to our intent. We consider only a small subset of possible restrictions 
on emergent software systems (Table \ref{TabPrm}), and our analysis, though 
complete with respect to some of these restrictions (Tables 
\ref{TabSPCA_ESCreate}--\ref{TabSPCA_ESAdapt2}), is still incomplete with 
respect to the whole subset, let alone the universe of possible restrictions.
The successes in real-world emergent software systems to date show that
tractability is indeed possible in some circumstances. The key issue now is to 
determine in detail those circumstances in which tractability does and does not hold.
Our results should thus be seen not as final statements but rather interim guidelines on 
how to
address this issue. We see the involvement of software engineers in this process
as essential (Section \ref{SectDiscCav}). To this end, we have tried to make
the reasoning used to derive our results (both in general (Section
\ref{SectResDesSolv}) and in our proofs) accessible to software 
engineers, to make plain to them the rather limited circumstances under which 
our results hold and thus enable them to break these results by suggesting
additional restrictions relative to which real-world emergent software system creation and
adaptation are provably tractable.

Before we close out this subsection, two issues with respect to the results 
listed above should be noted. First, though
certain notations (e.g., our conception of software system requirements) and some of the 
general ideas underlying proof techniques developed in \cite{WS15,WS22} are re-used in this paper,
none of the results for component-based software system creation derived in \cite{WS15,WS22} carry
over. This is because the restrictions on overall system structure and the number 
of types of components in a system critical to the proofs of those results in \cite{WS15,WS22} 
have no 
analogues in problems ESCreate and ESAdapt investigated here. Second, all results 
in this paper are derived relative to the classic Turing machine model of computation, and 
hence do not directly address issues of efficient solvability or unsolvability under other 
models of computation such as quantum computers. That being said, as will be discussed in 
Section \ref{SectDiscCav}, our results indirectly imply certain consequences for the efficient 
solvability of ESCreate and ESAdapt under such alternative models of computation.

\subsection{Organization of the Paper}

Our paper is organized as follows. In Section \ref{SectForm}, we summarize the 
emergent software system model given in \cite{FP17,PG+16} and formalize 
the problems of emergent software system creation and adaptation. In Section 
\ref{SectRes}, we first in Section \ref{SectResDesSolv} describe several 
popular conceptions of efficient solvability and then in Sections 
\ref{SectResESC} and \ref{SectResESA} assess the efficient solvability of 
emergent software system creation and adaptation, respectively, relative to each
of these conceptions. In order to focus in the main text on the implications of
our results, proofs of several of these results are given in an appendix. Our 
results are summarized and discussed in Section \ref{SectDisc}. Finally, our 
conclusions and directions for future work are given in Section \ref{SectConc}.

\section{Formalizing Emergent Software Creation and \\ Adaptation}

\label{SectForm}

In this section, we first review the basic entities in the model of emergent
software given in \cite{FP17,Por14,PG+16} --- namely, software system 
requirements, interfaces and components, component-based software systems, and 
emergent software systems. We then formalize two computational problems
associated with emergent software system creation and adaptation.

\begin{figure}[t]
\centering
\begin{tabular}{| c || c | c | c | c | c || c |}
\hline
req.\ & $x_1$ & $x_2$ & $x_3$ & $x_4$ & $x_5$ & output \\
\hline\hline
$r_1$ & T & T & T & T & T & $2$ \\
\hline
$r_2$ & T & F & F & F & T & $1$ \\
\hline
$r_3$ & F & F & F & F & F & $2$ \\
\hline
$r_4$ & F & F & F & F & F & $2$ \\
\hline
$r_5$ & T & T & T & F & T & $3$ \\
\hline
\end{tabular} \\

{\tt \small
\begin{tabular}{l l}
 & \\
\multicolumn{2}{c}{(a)} \\
 & \\
interface intSystem \{                      & interface intProc1 \{ \\
~ void systemMain(Input I)                  & ~ void callProc1(Input I) \\
\}                                          & \} \\
 & \\
interface intProc2 \{                       & interface intProc3 \{ \\
~ void callProc2(Input I)                   & ~ void callProc3(Input I) \\
\}                                           & \} \\
 & \\
interface intProc \{                         & \\
~ void proc(Input I)                        & \\
\}                                           & \\
\multicolumn{2}{c}{(b)} \\
\end{tabular}
}
\caption{An Example Emergent Software System (Adapted from
           \cite[Figure 1]{WS15}). 
           (a) Software requirements $R = \{r_1, r_2, r_3, r_4, r_5\}$ defined on
           Boolean variables $X = \{x_1, x_2, x_3, x_4, x_5\}$ and output-set
           $O = \{1, 2, 3\}$. (b) A sample component interface library $L_{int}$ consisting
           of five interfaces. 
}
\label{FigExESS1}
\end{figure}

\begin{figure}[p]
\centering
{\tt \small
\begin{tabular}{l l}
component System1                           & component System2 \\ 
\mbs provides intSystem                      & \mbs provides intSystem \\
\mbs requires intProc1, intProc2, intProc3 \{ \mbs & \mbs requires intProc1 \{ \\
\mss void systemMain(Input I) \{                & \mss void systemMain(Input I) \{ \\
\mbs if $v_I(x_1)$ then callProc1(I)         & \mbs callProc1(I) \\
\mbs elsif $v_I(x_5)$ then callProc2(I)      & \mbs \}\} \\
\mbs else callProc3(I)                       & \\
\mss \}\}                                         & \\
 & \\
component Proc1                             & component Proc2 \\
\mbs provides intProc1                       & \mbs provides intProc2 \\
\mbs requires intProc \{                      &\mbs requires intProc \{ \\
\mss void callProc1(Input I) \{                 & \mss void callProc2(I) \{ \\
\mbs proc(I)                                 & \mbs proc(I) \\
\mss \}\}                                         & \mss \}\} \\
 & \\
component Proc3                             & component Base \{ \\
\mbs provides intProc3                       & \mss void main(Input I) \{ \\
\mbs requires intProc \{                      & \mbs systemMain(Input I) \\
\mss void callProc3(Input I) \{                 & \mss \}\} \\
\mbs proc(I)                                 & \\
\mss \}\}                                         & \\
 & \\
component ProcA                             & component ProcB \\
\mbs provides intProc \{                      & \mbs provides intProc \{ \\
\mss void proc(Input I) \{                      & \mss  void proc(Input I) \{ \\
\mbs if $v_I(x_4)$ then output 2             & \mbs if not $v_I(x_2)$ then output 1 \\
\mbs elsif not $v_I(x_3)$ then output 1      & \mbs elsif not $v_I(x_4)$ then output 2 \\
\mbs elsif $v_I(x_5)$ then output 3          & \mbs else output 3 \\
\mbs else $a_1$                              & \mss \}\} \\
\mss \}\} \\                                      & \\
& \\
component ProcC                             & component ProcD \\
\mbs provides intProc \{                      & \mbs provides intProc \{ \\
\mss void proc(Input I) \{                      & \mss void proc(Input I) \\
\mbs if $v_I(x_4)$ then output 2             & \mbs output 2 \\
\mbs else output 2                           & \mss \}\} \\
\mss \}\}                                       & \\
& \\
\multicolumn{2}{c}{(c)} \\
\end{tabular}
}
\caption{ An Example Emergent Software System (Cont'd). 
           (c) A sample software component library $L_{comp}$
           consisting of ten components.
}
\label{FigExESS2}
\end{figure}

\begin{figure}[t]
\begin{center}
\includegraphics[width=5.0in]{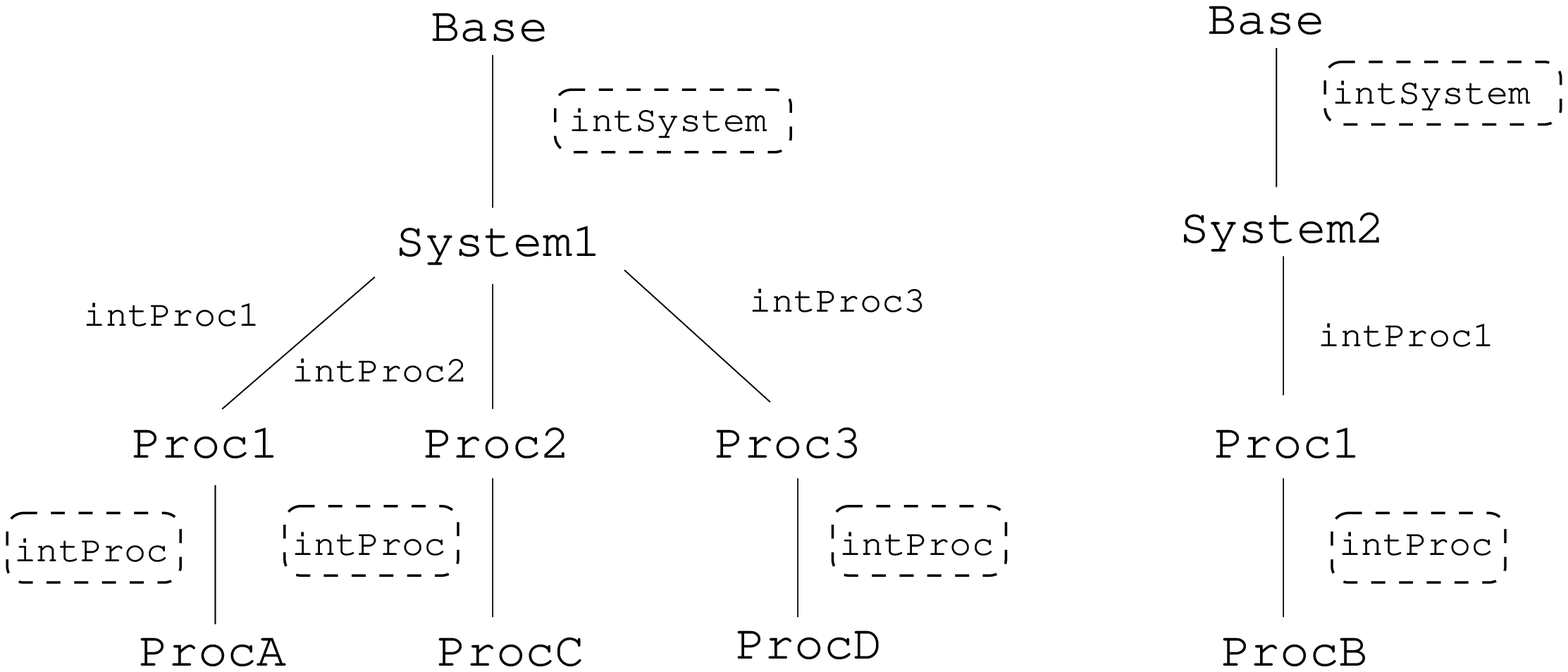}
\end{center}
\caption{ An Example Emergent Software System (Cont'd). 
           (d) Component wiring trees corresponding to two valid component-based
           software systems for $R$ in part (a) based on $L_{int}$, $L_{comp}$,
           and base component {\tt Base} $\in L_{comp}$ given in
           parts (b) and (c). Note that interfaces {\tt intSystem} and {\tt intProc} can each
           be implemented by multiple components (namely, {\tt System1} and {\tt System2} 
           ({\tt intSystem}) and {\tt ProcA}, {\tt ProcB}, {\tt ProcC}, and {\tt ProcD}
           ({\tt intProc})) and are hence key in allowing component choices yielding
           different software systems; this is acknowledged by putting dashed boxes
           around these interfaces in the component wiring trees diagrams.
}
\label{FigExESS3}
\end{figure}

The basic entities in our model are formalized as follows:

\begin{itemize}
\item{
\textbf{Software system requirements}: The requirements will be a set $R = 
\{r_1, r_2,$ \linebreak  $\ldots, r_{|R|}\}$ of input-output pairs where each 
pair $r_j = (i_j,o_j)$ consists of an input $i_j$ defined by a particular sequence 
of truth-values $i_j = \langle v_{i_j}(x_1), v_{i_j}(x_2), \ldots,$ \linebreak
$v_{i_j}(x_{|X|})\rangle$, 
$v_{i_j}(x_k) \in \{True, False\}$, relative to each of the Boolean variables $x_k$ 
in set $X = \{x_1, x_2, \ldots, x_{|X|}\}$ and an output $o_j$ from set $O = 
\{o_1, o_2, \ldots, o_{|O|}\}$. As such, these are functional requirements 
describing wanted system input-output behaviors and correspond to the 
pre-specified abstract goal of the system \cite[page 3]{PF19}.
An example set of software system requirements is given in part (a) of 
Figure \ref{FigExESS1}.
}
\item{
\textbf{Interfaces and Components}: We use the runtime component model 
underlying the Dana programming language as specified in \cite{Por14,PG+16}. In
particular, following \cite[Page 335]{PG+16}, let an \textbf{interface} be a set
of function prototypes, each comprising a function name, return type and 
parameter types, and a set of transfer fields, typed pieces of state that 
persist across alternate implementations of the interface during runtime 
adaptation. A \textbf{component} has one or more provided interfaces and zero or
more required interfaces; the component has implementations for all functions 
specified in the provided interfaces, and these implementations in turn call 
upon functions and transfer fields specified in the required interfaces. We will
assume that all available interfaces and components are stored in libraries 
$L_{int}$ and $L_{comp}$, respectively.  Example interface and component 
libraries relative to the software system requirements in part (a) of Figure \ref{FigExESS1}
are given shown in parts (b) and (c) of Figures \ref{FigExESS1} and \ref{FigExESS2},
respectively.

Note that the interfaces and components in Figures\ref{FigExESS1} and \ref{FigExESS2} (and indeed all subsequent
interfaces and components specified in this paper) are described using a notation
approximating that of the Dana programming language, e.g., \cite[Figure 2]{PG+16}. 
This is done to ensure both that the proofs of results given here do not invoke programming 
language features more powerful than those available in Dana and that these results are 
hence applicable to emergent software systems as described in \cite{FP17,Por14,PG+16}, which
are written in Dana. 
}
\item{
\textbf{Component-based software systems}: We use the model of component-based 
software systems specified in \cite{FP17,Por14,PG+16}. In particular, given 
interface and component libraries $L_{int}$ and $L_{comp}$ and a base component
$c \in L_{comp}$ implementing a $main$ function, a \textbf{valid component-based
software system} $S$ consists of a set of component-choices from $L_{comp}$ 
including $c$ that not only implements all required interfaces of $c$ but also 
recursively implements all required interfaces of those component-choices and 
their sub-component-choices, if any. The interface-connections between these 
components are called \textbf{wirings}. In order to allow data transfer
fields to hold different values relative to different implementations of
an interface by the same component, these implementations are done 
relative to copies of that component. Moreover, in cases where a 
component provides multiple interfaces, different implementations of
that component relative to two of those interfaces are done relative to
reduced copies of that component, both of which only contain code for and
thus provide only those services in the component specified by their
respective implementing interfaces. 

Any valid component-based software system has an associated directed vertex- and
arc-labeled tree in which the vertices are components, the arcs are the 
wirings, and the vertices and arcs are labeled with the names of the associated
components and interfaces from $L_{comp}$ and $L_{int}$, respectively; let this
tree be called the \textbf{component wiring tree $T$ associated with $S$}. 
The component wiring trees of two valid component-based software systems
relative to $R$, $L_{int}$, $L_{comp}$, and base component {\tt Base} $\in L_{comp}$ given in parts (a)--(c) of Figures
\ref{FigExESS1} and \ref{FigExESS2} are given in Figure \ref{FigExESS3}. 
Note that such a $T$ has a single root vertex (namely, base component $c$), 
exactly one directed path from this root to any vertex,
and is labeled such that a component vertex-label does not occur more
than once on any directed path from the root to a leaf, i.e., there
are no recursive dependencies \cite[page 10]{FP17}. This prevents the 
possibility of infinite-depth software systems resulting from the 
interface-component implementation sequence between two same-label components
on such a path being repeated an infinite number of times. 

Given a set $R$ 
of software system requirements, a valid component-based software system $S$ is
a \textbf{working component-based software system relative to $R$} if for each 
input-output pair $(i_j, o_j) \in R$, the output of $S$ given input $i_j$ is $o_j$. 
For example, the software system on the left in Figure \ref{FigExESS3}
satisfies all requirements in $R$ given in Figure\ref{FigExESS1}(a)  and hence is
a working component-based software system relative to $R$, but the software system
on the right is not (because it produces different outputs (3, 1, 1, and 2, respectively)
for requirements $r_1$, $r_3$, $r_4$, and $r_5$).
}
\item{
\textbf{Emergent software system}: As defined in \cite{FP17,PG+16}, an
\textbf{emergent \linebreak (component-based) software system} is one that, given initial 
functional software systems requirements $R$, interface and component libraries
$L_{int}$ and $L_{comp}$, and a base component $c \in L_{comp}$, self-assembles and 
self-adapts as necessary to optimize system performance as its running environment changes 
over time. For a software system $S$, a system's running environment and 
performance are quantified in terms of events and metrics whose values are sampled at 
discrete times during system operation \cite[pages 10--11]{FP17}; aspects of system
performance used by a learning algorithm to guide both self-assembly and
self-adaptation are in turn summarized in a reward function. As currently 
implemented \cite{FP17},
the initial self-assembly phase creates a list of all working software systems relative
to the given $L_{int}$, $L_{comp}$, and $R$, where each system is described by a unique
ID string that lists all components in the system and their interconnections. During
the subsequent self-adaptation phase, the learning algorithm searches over this list to find
appropriate alternatives to the currently-running system that might help optimize system 
performance \cite[Section 3.2.1]{FP17}.

In this paper, we
will assume that a system's running environment and performance are sampled using
an environment function $Env()$ that maps a given system onto a collection of
events and metric-values and that the reward function $Rew()$ maps the latest
accumulated performance metric values for a system onto a single positive integer value. For
simplicity, we shall further assume that $Env()$ and $Rew()$ are computable in time 
polynomial in the size of $S$, $Rew(S)$ is optimized by minimization, i.e., smaller values 
of $Rew(S)$ are preferred, and $Rew()$ is used to choose among but does not alter the
set of possible working software systems relative to a given $L_{int}$, $L_{comp}$,
and $R$.
}
\end{itemize}

We can now formalize computational problems corresponding to emergent 
software adaptation as conceived in \cite{FP17,PG+16}:

\vspace*{0.1in}

\noindent
{\sc Emergent Software Creation } (ESCreate)  \\
{\em Input}: Software system requirements $R$, interface and component libraries
              $L_{int}$ and $L_{comp}$, a base component $c \in L_{comp}$, and
              reward and environment functions $Rew()$ and $Env()$.\\
{\em Output}: A working component-based software system $S$ based on $c$ relative to 
               $L_{int}$, $L_{comp}$, and $R$ that has the 
               smallest value of $Rew(S)$ over all working systems based on $c$
               relative to $L_{int}$, $L_{comp}$, and $R$, if any working system 
               exists, and special symbol $\bot$ otherwise.

\vspace*{0.1in}

\noindent
{\sc Emergent Software Adaptation } (ESAdapt)  \\
{\em Input}: Software system requirements $R$, interface and component libraries
              $L_{int}$ and $L_{comp}$, a working component-based software 
              system $S$ based on component $c \in L_{comp}$ relative to $R$, 
              $L_{int}$, and $L_{comp}$, and reward and environment functions 
              $Rew()$ and $Env()$.\\
{\em Output}: A working component-based software system $S'$ based on $c$ 
               relative to $L_{int}$, $L_{comp}$, and $R$ that has the
               smallest possible value of $Rew(S')$ over all working systems based on $c$
               relative to $L_{int}$, $L_{comp}$, and $R$.

\vspace*{0.1in}

\noindent
Problem ESCreate corresponds to the initial creation of a working emergent
software system, while ESAdapt corresponds to each subsequent modification of the system to 
optimize the reward function as the system's running environment changes over time. 
These problems do not correspond directly to emergent system creation and 
adaptation as currently implemented, in that ESCreate only returns one rather than all
working software systems relative to $L_{int}$, $L_{comp}$, $R$, ESCreate optimizes
the performance of this returned system relative to $Rew()$, and ESAdapt only has access to
$S$ and not the complete list of working software systems relative to $L_{int}$,
$L_{comp}$, and $R$.\footnote{
An additional difference and possible source of confusion is that the term ``adaptation''
in \cite{FP17} refers to the process of changing the components and interconnections in the 
currently-running software system to those in the candidate selected by the learning 
algorithm. In this paper, we shall instead use ``adaptation'' to refer to the overall 
self-adaptation process process described by ESAdapt, which may involve multiple episodes 
of adaptation in the sense of \cite{FP17} to achieve an optimal system under $Rew()$. 
}
However, as will be discussed in Section \ref{SectDiscImpRW}, complexity
results for ESCreate and ESAdapt can still be used to investigate both current and potential 
implementations of emergent software systems.

Two additional notes are in order about our definitions of ESCreate and ESAdapt.
First, the given $L_{int}$, and $L_{comp}$ in an
input of ESCreate may not allow a working software system relative to the given $R$
but there is always a working software system for any input of ESAdapt --- namely, $S$.
Second, $Rew()$ and $Env()$ are part of the input for both ESCreate and ESAdapt 
and must be included in any instance of these problems, i.e., it cannot be the case that
$Env()$ and/or $Rew()$ are empty. That being said, $Env()$ (as well as $S$
in the case of ESAdapt) need not necessarily be used by any algorithm solving these 
problems but are provided as part of the problem inputs to make results derived here
relevant to emergent systems as described in \cite{FP17,PG+16} (see Section \ref{SectDisc}
for further discussion on the latter point).

In this paper, we will consider the following 
versions of $Rew()$:

\begin{itemize}
\item $Rew_{\#comp}(S) =$ the number of components in $S$.
\item $Rew_{CodeB}(S) =$ the size of the codebase of $S$, i.e., the
                        total number of lines of code in the 
                        interfaces and components comprising $S$.
\end{itemize}

\noindent
Relative to a particular reward function $Rew_X()$, we refer to our problems 
above as ESCreate under $Rew_X()$ and ESAdapt under $Rew_X()$, respectively.
Note that results derived relative to $Rew_{\#comp}()$ and $Rew_{CodeB}()$ have 
broad applicability as the values of these functions correlate with the values 
of at least some of the reward functions studied to date in emergent software 
systems, e.g., system response time \cite[page 11]{FP17}. As none of our
result proofs rely on $Env()$, we need not specify its form further.

Let us now illustrate problems ESCreate and ESAdapt relative to the example emergent 
software system described in Figures \ref{FigExESS1}--\ref{FigExESS3}:

\begin{itemize}
\item{
{\bf Valid Software Systems relative to $L_{int}$ and $L_{comp}$}:
As each interface {\tt proc1}, {\tt proc2}, and {\tt proc3} in software systems
based on component {\tt System1} can be implemented with any of the components
{\tt ProcA}, {\tt ProcB}, {\tt ProcC}, or {\tt ProcD}, there are $4^3 = 64$ valid software 
systems relative to the general structure on the left in part (d) of Figure \ref{FigExESS3}.
By analogous reasoning relative to component {\bf System2} and interface {\tt proc1},
there are 4 valid software systems relative to the general structure on the right. Hence,
there are in total 68 valid software systems relative to $L_{int}$ and $L_{comp}$.
}
\item{
{\bf Working Software Systems relative to $R$}:
The reader can verify that
(i) there are no working software systems relative to $R$ that incorporate component 
{\tt System2} and (ii) all working software systems relative to $R$ that incorporate
component {\tt System1} must implement interface {\tt Proc1} with component {\tt ProcA}
and interface {\tt Proc3} with either of the components {\tt ProcC} or {\tt ProcD}, and
that interface {\tt Proc2} (as its implementing component code is never executed) can
be implemented by any of the components {\tt ProcA, ProcB, ProcC} or {\tt ProcD}.
Thus, the implementations of {\tt Proc1}, {\tt Proc2}, and {\tt Proc3} in
{\tt System1} that yield working software systems are
{\tt (ProcA, ProcA, ProcC)}, 
{\tt (ProcA, ProcA, ProcD)}, 
{\tt (ProcA, ProcB, ProcC)}, 
{\tt (ProcA, ProcB, ProcD)}, 
{\tt (ProcA, ProcC, ProcC)}, 
{\tt (ProcA, ProcC, ProcD)}, 
{\tt (ProcA, ProcD, ProcC)}, and
{\tt (ProcA, ProcD, ProcD)}, 
giving in total 8 working software system relative to $R$.
}
\item{
{\bf Output of ESCreate and ESAdapt under different $Rew()$}: As all working software systems
relative to $R$
have 8 components, ESCreate under $Rew_{\#comp}()$ can return any of them;
however, as the system with the smallest codebase implements both {\tt Proc2} and
{\tt Proc3} with {\tt ProcC}, it is this system that would be selected by 
ESCreate under $Rew_{CodeB}()$. By analogous reasoning, ESAdapt under 
$Rew_{\#comp}$ can return any of the 8 working software systems given system $S$ 
drawn from the same set; however, ESAdapt under $Rew_{CodeB}()$ must (regardless of the 
choice of the given $S$) return the system that implements both {\tt Proc2} and 
{\tt Proc3} with {\tt ProcC}.
}
\end{itemize}

\noindent
This concludes our formalization of emergent software system creation and adaptation. 
A reasonable conjecture at this point is that
ESAdapt will be easier to solve than ESCreate, given that the former is given a
working system as part of its input. This will be assessed below.

\section{Results}

\label{SectRes}

In this section, we will use computational complexity analysis to
assess viable algorithmic options for efficient emergent software system
creation and adaptation. This will be done relative to various
desirable types of efficient solvability described in Section
\ref{SectResDesSolv}. The results of our analyses for emergent
software creation and adaptation are given in Sections \ref{SectResESC}
and \ref{SectResESA}, respectively. 

It turns out that both ESCreate and ESAdapt are unsolvable in the most
general possible case --- that is, neither ESCreate nor ESAdapt have algorithms
that always return the correct output for an input incorporating any possible choices of
$Env()$ and $Rew()$ (in the case of ESCreate) or any possible choice of $Env()$ and
either of $Rew_{\#comp}()$ or $Rew_{CodeB}()$ (in the case of ESAdapt) and in which there 
are no restrictions on the form, size, or running times of $L_{int}$ and $L_{comp}$, their 
member interfaces and components, or any software systems created using $L_{int}$ and 
$L_{comp}$ (see Sections \ref{SectResESCUnsolv} and \ref{SectResESAUnsolv},
respectively). Hence, the remainder of our analyses will be done relative to restricted 
versions of these problems in which candidate component-based software systems run and 
hence can be verified against given system requirements in polynomial time.

\subsection{Types of Efficient Solvability}

\label{SectResDesSolv}

Consider the following desirable forms of solvability:

\begin{enumerate}
\item{
\textbf{Polynomial-time exact solvability}: An exact 
polynomial-time algorithm is a deterministic algorithm whose runtime is 
upper-bounded by $c_1|x|^{c_2}$, where $|x|$ is the size of the input $x$
and where $c_1$ and $c_2$ are constants, and
is always guaranteed to produce the correct output for all inputs. A problem 
that has such an algorithm is said to be \textbf{polynomial-time 
tractable}. Polynomial-time tractability is desirable because runtimes increase
slowly as input size increases, and hence allow the solution of larger inputs.

It is possible that the computational difficulty of a problem may be inflated
 in general by inputs that have no solutions, and hence force any algorithm to
exhaustively consider all possible candidate solutions. In such cases, it is 
useful to assess whether a problem is \textbf{polynomial-time exact promise 
solvable} --- that is, whether that problem is exactly solvable in polynomial 
time on those inputs which are guaranteed to have solutions, where these
guarantees are known as \textbf{promises}.
}
\item{
\textbf{Polynomial-time approximate solvability}: A polynomial-time 
approximation algorithm is an algorithm that runs in polynomial time in an 
approximately correct (but acceptable) manner for all inputs. There are a number
of ways in which an algorithm can operate in an approximately correct manner. 
Three of the most popular ways are as follows:

\begin{enumerate}
\item {\bf Frequently Correct (Deterministic)} \cite{HW12}: Such an algorithm 
       runs in polynomial time and gives correct solutions for all but a very 
       small number of inputs. In particular, if the number of inputs for each 
       input-size $n$ on which the algorithm gives the wrong or no answer 
       (denoted by the function $err(n)$) is sufficiently small (e.g., $err(n) 
       = c$ for some constant $c$), such algorithms may be acceptable.
\item {\bf Frequently Correct (Probabilistic)} \cite{MR10}: Such an algorithm 
       (which is typically probabilistic) runs in polynomial time and gives 
       correct solutions with high probability.  In particular, if the 
       probability of correctness is $\geq 2/3$ (and hence can be boosted by 
       additional computations running in polynomial time to be correct with 
       probability arbitrarily close to 1 \cite[Section 5.2]{Wig07}), such 
       algorithms may be acceptable.
\item {\bf Approximately Optimal} \cite{AC+99}: Such an algorithm $A$ runs in
       polynomial time and gives a solution $A(x)$ for an input $x$ whose value
       $val(A(x))$ is guaranteed to be within a multiplicative factor $f(|x|)$ of
       the value $v_{OPT}(x)$ of an optimal solution for $x$, i.e., 
       $|v_{OPT}(x) - val(A(x))| \leq f(|x|) \times v_{OPT}(x)$ for any input $x$
       for some function $f()$. A problem with such an algorithm is said to be 
       polynomial-time $f(|x|)$-approximable. In particular, if $f(|x|)$ is a 
       constant very close to 0 (meaning that the algorithm is always
       guaranteed to give a solution that is either optimal or very close to 
       optimal), such algorithms may be acceptable.
\end{enumerate}
}
\item{
\textbf{Effectively polynomial-time exact restricted solvability}: Even if
a problem is not solvable in any of the senses above, a restricted version
of that problem may be exactly solvable in close-to-polynomial time.
Let us characterize restrictions on problem inputs in terms of a 
set $K = \{k_1, k_2, \ldots, k_{|K|} \}$ of aspects of the input. 
For example, possible restrictions on the inputs ESCreate could
be the number of given software requirements, the number of components
in $L_{comp}$, and the maximum number of components in a working
software system relative to $L_{int}$, $L_{comp}$, and $R$ (see also
Table \ref{TabPrm} in Section \ref{SectResESC_FPSolv}). Let each such
aspect be called a {\bf parameter}.

One of the most
popular ways in which an algorithm can operate in close-to-polynomial
time relative to restricted inputs is {\bf fixed-parameter (fp-) \linebreak tractability} 
\cite{DF99}. Such an algorithm runs in time that is non-polynomial purely in 
terms of the parameters in $K$, {\em i.e.}, in time $f(K)|x|^c$ where $f()$ is 
some function, $|x|$ is the size of input $x$, and $c$ is a constant. A problem
with such an algorithm for parameter-set $K$ is said to be \textbf{fixed-parameter
(fp-)tractable relative to $K$}. Fixed-parameter tractability
generalize polynomial-time exact solvability by allowing the leading constant 
$c_1$ of the input size in the runtime upper-bound of an algorithm to be a 
function of $K$. Though such algorithms run in non-polynomial time in general, 
for inputs in which all the parameters in $K$ have very small constant values 
and $f(K)$ thus collapses to a possibly large but nonetheless constant 
value, such algorithms (particularly if $f()$ is suitably well-behaved,
({\em e.g}, $(1.2)^{k_1 + k_2}$) may be acceptable.
}
\end{enumerate}

In the following two subsections, we shall evaluate the algorithmic options
for ESCreate and ESAdapt, respectively,  relative to each of these types of 
solvability. Our unsolvability proofs will use reductions between pairs
of problems, where a reduction from a problem $\Pi$ to a problem $\Pi'$ is
essentially an efficient algorithm $A$ for solving $\Pi$ which uses a 
hypothetical algorithm for solving $\Pi'$. Reductions are useful by the
following logic:
 
\begin{itemize}
\item If $\Pi$ reduces to $\Pi'$ and $\Pi'$ is efficiently solvable by
       algorithm $B$ then $\Pi$ is efficiently solvable (courtesy of
       the algorithm $A'$ that invokes $A$ relative to $B$).
\item If $\Pi$ reduces to $\Pi'$ and $\Pi$ is not efficiently solvable 
       then $\Pi'$ is not efficiently solvable (as otherwise, by the logic
       above, $\Pi$ would be efficiently solvable, which would be a
       contradiction).
\end{itemize}

\noindent
We will use the following three types of reducibility:

\begin{definition}
\cite[Section 3.1.2]{Gol08}
Given decision problems $\Pi$ and $\Pi'$, i.e., problems whose answers
are either ``Yes'' or ``No'', $\Pi$ {\em polynomial-time (Karp) reduces
to} $\Pi'$ if there is a polynomial-time computable function $f()$ such that for
any instance $x$ of $\Pi$, the answer to $\Pi$ for $x$ is ``Yes'' if and only 
if the answer to $\Pi'$ for $f(x)$ is ``Yes''.
\end{definition}

\begin{definition}
\cite[Section 3.1.2]{Gol08}
Given search problems $\Pi$ and $\Pi'$, i.e., problems whose answers
are actual solutions rather than just  ``Yes'' or ``No'', $\Pi$ 
{\em polynomial-time (Levin) reduces to} $\Pi'$ if there is a pair of 
polynomial-time functions $f()$ and $g()$ such that for any instance $x$ of 
$\Pi$, the answer to $\Pi$ for $x$ is $g(x,y)$ if and only if the answer to
$\Pi'$ for $f(x)$ is $y$.
\end{definition}

\begin{definition}
\cite{DF99}\footnote{
Note that this definition given here is actually Definition 6.1 in 
\cite{vRB+19}, which modifies that in \cite{DF99} to accommodate
parameterized problems with multi-parameter sets.
}
Given parameterized decision problems $\Pi$ and $\Pi'$, $\Pi$ {\em parameterized
reduces to} $\Pi'$ if there is a function $f()$ which transforms instances
$\langle x, K \rangle$ of $\Pi$ into instances $\langle x', K' \rangle$ of
$\Pi'$ such that $f()$ runs in $f'(K)|x|^c$ time for some function $f'()$ and
constant $c$, $k' = g_{k'}(K)$ for each $k' \in K$ for some function $g_{k'}()$,
and for any instance $\langle x, K \rangle$ of $\Pi$, the answer to $\Pi$ for 
$\langle x, K\rangle$ is ``Yes'' if and only if the answer to $\Pi'$ for 
$f(\langle x, K \rangle)$ is ``Yes''.
\end{definition}

\noindent
Our reductions will be from versions of the following problems:

\vspace*{0.1in}

\noindent
{\sc Turing Machine Halting} (TM Halting) \cite{Tur36}  \\
{\em Input}: A Turing Machine $M$ and a binary string $x$. \\
{\em Question}: Does $M$ halt when given $x$ as input?

\vspace*{0.1in}

\noindent
{\sc Dominating set} \cite[Problem GT2]{GJ79} \\
{\em Input}: An undirected graph $G = (V, E)$ and a positive integer $k$. \\
{\em Question}: Does $G$ contain a dominating set of size $k$,
              i.e., is there a subset $V' \subseteq V$, $|V'| = k$, such
              that for all $v \in V$, either $v \in V'$ or there is at least one
              $v' \in V'$ such that $(v, v') \in E$?

\vspace*{0.1in}

\noindent
{\sc Optimal Dominating set} ({\sc Dominating set}$^{OPT}$) \\
{\em Input}: An undirected graph $G = (V, E)$. \\
{\em Output}: A dominating set in $G$ of minimum size.

\vspace*{0.1in}

\noindent
For each vertex $v \in V$ in a graph $G$, let the complete neighbourhood 
$N_C(v)$ of $v$ be the set composed of $v$ and the set of all vertices in $G$ 
that are adjacent to $v$ by a single edge, i.e., $v \cup \{ u ~ | ~ u ~ \in 
V ~ \rm{and} ~ (u,v) \in E\}$. We assume below for each instance of 
{\sc Dominating set} an arbitrary ordering on the vertices of $V$ such that
$V = \{v_1, v_2, \ldots, v_{|V|}\}$. Note that only the first of the three
problems above is provably unsolvable (indeed, unsolvable in the sense that 
there can be no algorithm period that returns the correct output for every 
input \cite[Section 9.2.4]{HMU01}). Versions of the others are only known to be
unsolvable relative to the types of efficient solvability listed at the start 
of this subsection modulo the conjectures $P \neq NP$ and 
$FPT \neq W[2]$; however, this is not a problem in practice as both of 
these conjectures are widely believed within computer science to be true 
\cite{DF13,For09}.

As we shall often see in the following two sections, a single reduction may
imply multiple results. For example, with respect to the third of the
solvability options described above, additional and 
sometimes stronger fp-tractability and intractability results can
often be derived using the following three lemmas.

\begin{lemma}
\cite[Lemma 2.1.30]{War99}
If problem $\Pi$ is fp-tractable relative to parameter-set $K$ then $\Pi$ is
       fp-tractable for any parameter-set $K'$ such that $K \subset K'$.
\label{LemPrmProp1}
\end{lemma}

\begin{lemma}
\cite[Lemma 2.1.31]{War99}
If problem $\Pi$ is fp-intractable relative to parameter-set $K$ then $\Pi$ is
       fp-intractable for any parameter-set $K'$ such that $K' \subset K$.
\label{LemPrmProp2}
\end{lemma}

\begin{lemma}
\cite[Lemma 2.1.35]{War99}.
If problem is $NP$-hard when all parameters in \linebreak parameter-set $K$ have constant
values then $\Pi$ cannot be fp-tractable relative to any subset of $K$ unless
$P = NP$.
\label{LemPrmProp3}
\end{lemma}

\noindent
There are a variety of techniques for creating a reduction from a problem $\Pi$ to
a problem $\Pi'$ (\cite[Section 3.2]{GJ79}; see also \cite[Chapters 3 and 6]{vRB+19}). 
One of these techniques is component design, in which an instance of $\Pi'$ constructed
by a reduction is 
structured as mechanisms that generate candidate solutions for the given instance
of $\Pi$ and check these candidates to see if any are actual solutions. We have already
seen in the example software systems given in Figure \ref{FigExESS3} how interfaces
with different implementing components (in that case, interfaces {\tt intSystem} and
{\tt intProc}) can be used to generate choices when constructing a component-based software 
system.  In subsequent subsections, we will use this and other features of interfaces and 
components under the Dana runtime model as described in Section \ref{SectForm} 
to structure mechanisms that generate candidate solutions (i.e., valid component-based
software systems corresponding to vertex-sets of size $k$ in a given graph $G$) and check 
these candidates to see if they are actual solutions (e.g., working component-based software
systems relative $R$ corresponding to dominating sets of size $k$ in $G$) in many of the 
reductions underlying our results for problems ESCreate and ESAdapt.

\subsection{Results for Emergent Software Creation}

\label{SectResESC}

Many of the results derived in this section for ESCreate will actually be 
derived relative to the following problem:

\vspace*{0.1in}

\noindent
{\sc Component-based Software Creation } (CSCreate)  \\
{\em Input}: Software system requirements $R$, interface and component libraries
              $L_{int}$ and $L_{comp}$, and a base component $c \in 
              L_{comp}$. \\
{\em Question}: Is there a working component-based software system $S$ based on
              $c$ relative to $L_{int}$, $L_{comp}$, and $R$?

\vspace*{0.1in}

\noindent
Note that each input for ESCreate has a corresponding input to CSCreate (namely, the 
input to ESCreate without $Rew()$ and $Env()$). Moreover, any algorithm $A$ 
that solves ESCreate under some $Rew()$ can be also used to solve CSCreate (namely,
if $A$ run on the given input $x$ for CSCreate produces a working system, output ``Yes'', 
otherwise output ``No''). This yields the following useful observation.

\begin{observation}
For any choice of $Rew()$ and $Env()$, if there is an algorithm $A$ of solvability 
type $T$ for ESCreate under $Rew()$ than
there is an algorithm $A'$ of solvability type $T$ for CSCreate.
\label{ObsESCreateSolvability}
\end{observation}

\subsubsection{Unsolvability of Unrestricted Emergent Software Creation}

\label{SectResESCUnsolv}

We start off by considering if problem ESCreate is solvable in the most
general possible case --- that is, if ESCreate has an algorithm that always
returns the correct output for an input incorporating any possible choices of
$Env()$ and $Rew()$ and in which there are no restrictions on the form, size,
or running times of $L_{int}$ and $L_{comp}$, their member interfaces and
components, or any software systems created using $L_{int}$ and $L_{comp}$. It 
turns out that such an algorithm cannot exist.

\begin{description}
\item[{\bf Result A.1}] For any choice of $Rew()$ and $Env()$, ESCreate is unsolvable.
\end{description}
\begin{proof}
Consider the following polynomial-time Karp reduction from {\sc TM Halting} to
CSCreate: given an instance $I = \la M, x\ra$ of {\sc TM Halting}, construct 
an instance $I' = \la R, L_{int}, L_{comp}, c\ra$ of CSCreate in
which $X = \{x_1\}$ and $O = \{1\}$, there is a single input-output pair $r$ in 
$R$ such that for $r = (True, 1)$, $L_{int}$ consists of the single interface

\begin{center}
\begin{verbatim}
    interface base {
         void main(Input I)
    }
\end{verbatim}
\end{center}

\noindent
and $L_{comp}$ consists of the single component

\begin{center}
\begin{verbatim}
    component Base provides base {
         void main(Input I) {
             <CODEM(x)>
             output 1
         }
    }
\end{verbatim}
\end{center}

\noindent 
where {\tt <CODEM(x)>} is the Dana code simulating the computation of
$M$ on input $x$. As Dana contains both loops and conditional statements,
it can readily simulate $M$ on input $x$ using code that is of size polynomial
in the sizes of the given descriptions of $M$ and $x$. Finally, let $c$ be 
component {\tt Base} in $L_{comp}$. Note that the instance of CSCreate
described above can be constructed in time polynomial in the size of the
given instance of {\sc TM Halting}. To conclude the proof, observe that the
only possible component-based system for the constructed instance of CSCreate
based on $c$ is that consisting of {\tt Base} itself, and that this system 
satisfies the sole input-output constraint in $R$ if and only if $M$ halts on 
input $x$ for the given instance of {\sc TM Halting}.
It is known that {\sc TM Halting} cannot have an algorithm that is correct for
all possible $\langle M, x\rangle$ instances \cite[Section 9.2.4]{HMU01}, and 
hence is unsolvable. Hence, the reduction above implies in turn that CSCreate
cannot have an algorithm either. The unsolvability result for ESCreate then 
follows by contradiction from Observation \ref{ObsESCreateSolvability}.
\end{proof}

\vspace*{0.1in}

\noindent
This result is especially disconcerting as it holds relative to not just
some choices but {\tt every} possible
choice of $Env()$ and $Rew()$ (this is because  the proof of this result ignores these
functions entirely). However, it is ultimately not surprising, given the
computational power inherent in the Dana programming language and the 
folklore result that a number of problems in software engineering, e.g.,
checking if a software system satisfies a set of given requirements, are known 
to be unsolvable as a consequence of Rice's Theorem \cite[Section 9.3.3]{HMU01}.

That being said, restricted versions of ESCreate may yet have correct and
even efficient algorithms. One reasonable such restriction is that any candidate
component-based software system $S$ created from a given $L_{int}$ and 
$L_{comp}$ runs in time polynomial in the input size $|I|$ and hence can be 
checked against the system requirements in $R$ in time polynomial in the size of
$R$ (as $|I| < |R|$), i.e., created software systems not only operate but can 
also be verified quickly. Indeed, such a restriction is implicit in the
requirement that emergent software systems be autonomously verifiable at
runtime \cite[page 5]{FP17}. In the remainder of our analyses in this paper, 
we will assume ESCreate and CSCreate to be so restricted, and 
will denote these restricted versions as ESCreate$_{poly}$ and CSCreate$_{poly}$, 
respectively.
 
\subsubsection{Polynomial-time Exact Solvability of Restricted Emergent Software Creation}

\label{SectResESC_PTExactSolv}

We now consider if ESCreate$_{poly}$ is efficiently solvable in the first of the senses listed at the
start of Section \ref{SectResDesSolv} --- namely, polynomial-time exact solvability and
polynomial-time exact promise solvability.
One might initially think that, given the somewhat radical nature of the
restriction on ESCreate proposed at the end of the previous subsection, 
ESCreate so restricted is now efficiently solvable in both of these senses. 
However, this turns out not to the case. 

These intractability results are shown using the following reduction. This reduction
creates valid component-based systems with component wiring trees of the form shown in Figure
\ref{FigRed2} in which the
multiply implemented interfaces {\tt cond1, cond2, \ldots condk} are used
to create valid software systems corresponding to all possible
vertex-sets of size $k$ in the graph $G$ in the given instance of
{\sc Dominating set}. As each input-output pair in the constructed $R$ corresponds to a 
vertex-neighbourhood in $G$, the code in component {\tt Base} ensures that working software 
systems correspond to dominating sets of size $k$ in $G$. 

\begin{figure}[t]
\begin{center}
\includegraphics[width=3.5in]{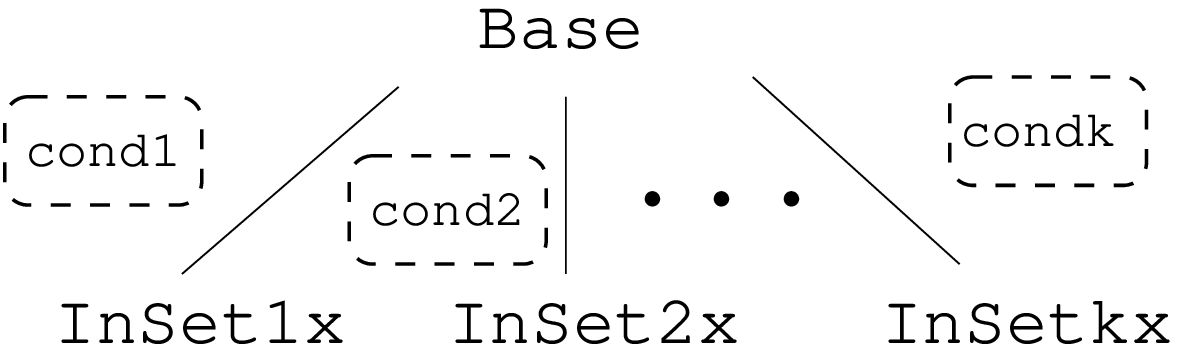}
\end{center}
\caption{General structure of valid software systems created by the reduction 
          in the proof of Lemma \ref{LemRed_DS_CSCreate1}. Note that index $x$ in
          {\tt InSetjx}, $ 1 \leq j \leq k$, is such that $1 \leq x \leq |V|$.
          Following the convention
          in Figure \ref{FigExESS3}, interfaces with multiple implementing
          components are enclosed in dashed boxes.
}
\label{FigRed2}
\end{figure}

\begin{lemma}
{\sc Dominating set} polynomial-time Karp reduces to CSCreate$_{poly}$.
\label{LemRed_DS_CSCreate1}
\end{lemma}

\begin{proof}
Given an instance $I = \la G = (V,E), k\ra$ of {\sc Dominating set}, construct 
the following instance $I' = \la R, L_{int}, L_{comp}, c\ra$ of CSCreate$_{poly}$: Let 
$X = \{x_1, x_2, \ldots, x_{|V|}\}$, i.e., there is a unique Boolean variable 
corresponding to each vertex in $V$, and $O = \{0, 1\}$. There
are $|V|$ input-output pairs in $R$ such that for $r_j = 
(i_j, o_j)$, $1 \leq j \leq |V|$, $v_{i_j}(x_k) = True$ if $v_k \in N_C(v_j)$
and is $False$ otherwise and $o_j = 1$.
Let $L_{int}$ consist of $k + 1$ interfaces broken into two groups:

\begin{enumerate}
\item{
A single interface of the form

\begin{center}
\begin{verbatim}
    interface base {
         void main(Input I)
    }
\end{verbatim}
\end{center}
}
\item{
A set of $k$ interfaces of the form

\begin{center}
\begin{verbatim}
    interface condJ {
         Boolean inSetJ(Input I)
    }
\end{verbatim}
\end{center}

\noindent
for $1 \leq {\tt J} \leq k$.
}
\end{enumerate}

\noindent
Let $L_{comp}$ consist of $k|V| + 1$ components broken into two groups:

\begin{enumerate}
\item{
A single component of the form

\begin{center}
\begin{verbatim}
    component Base provides base 
                   requires cond1, cond2, ..., condk {
         void main(Input I) {
             if inSet1(I) then output 1
             elsif inSet2(I) then output 1
                 ...
             elsif inSetk(I) then output 1
             else output 0
         }
    }
\end{verbatim}
\end{center}
}
\item{
A set of $k|V|$ components of the form

\begin{center}
\begin{verbatim}
    component InSetJK provides condJ {
         Boolean inSetJ(Input I) {
             return v_I(x_K) 
         }
    }
\end{verbatim}
\end{center}

\noindent
for $1 \leq {\tt J} \leq k$ and $1 \leq {\tt K} \leq |V|$.
}
\end{enumerate}

\noindent
Note that in $L_{comp}$, there are $|V|$ implementations of each
{\tt cond}-interface.  Finally, let $c$ be component {\tt Base} in $L_{comp}$.
Note that the instance of CSCreate$_{poly}$ described above can be constructed in 
time polynomial in the size of the given instance of {\sc Dominating set};
moreover, as there is only a $(k + 1)$-clause {\tt if-then} statement block and
no loops in the component code and $k \leq |V| < |I|$, any candidate
component-based software system created relative to $L_{int}$, $L_{comp}$, and $c$
runs in time linear in the size of input $I'$.

Let us now verify the correctness of this reduction:

\begin{itemize}
\item{
Suppose that there is a dominating set $D$ of size at most $k$ in the given instance
of {\sc Dominating set}. We can then construct a component-based software system
consisting 
of $c$ and the $|D|$ {\tt InSet}-components corresponding to the vertices
in $D$; the choice of which interface to implement for each vertex is
immaterial, and if there are less than $k$ vertices in $D$, the final
$k - |D|$ required {\tt cond}-interfaces can be implemented relative
to {\tt InSet}-components corresponding to arbitrary vertices in $D$. Observe that for 
each $(i_j,o_j) \in R$, this software system produces output $o_j$ given input $i_j$.
}
\item{
Conversely, suppose that the constructed instance of CSCreate$_{poly}$ has a 
working component-based software system based on $c$ relative to $L_{int}$,
$L_{comp}$, and $R$.  In order to correctly accommodate all 
input-output pairs in $R$, the $k$ {\tt if-then} statements in $c$ must 
implement {\tt InSet}-components whose corresponding vertices form a dominating 
set in $G$ of size at most $k$.
Hence, the existence of a working component-based software system for the
constructed instance of CSCreate$_{poly}$ implies the existence of a
dominating set of size at most $k$ for the given instance  of
{\sc Dominating set}.
}
\end{itemize}

\noindent
This completes the proof.
\end{proof}

\begin{description}
\item[{\bf Result A.2}] For any choice of $Rew()$ and $Env()$, if ESCreate$_{poly}$ is
polynomial-time exact solvable then $P = NP$.
\end{description}
\begin{proof}
Given the $NP$-hardness of {\sc Dominating set}, the reduction in Lemma
\ref{LemRed_DS_CSCreate1} implies that CSCreate$_{poly}$ is $NP$-hard, and hence not 
solvable in polynomial time unless $P = NP$. The polynomial-time
intractability result for ESCreate$_{poly}$
then follows by contradiction from Observation \ref{ObsESCreateSolvability}.
\end{proof}

\begin{description}
\item[{\bf Result A.3}] For any choice of $Rew()$ and $Env()$, if ESCreate$_{poly}$ is
polynomial-time exact promise solvable then $P = NP$.
\end{description}
\begin{proof}
Suppose that for some choice of $Rew()$ and $Env()$, ESCreate$_{poly}$ is 
polynomial-time promise solvable by an algorithm $A$.\footnote{
It may initially seem puzzling why, in light of Observation
\ref{ObsESCreateSolvability}, we here directly evaluate the polynomial-time
promise solvability of ESCreate$_{poly}$. This is necessary because the promise 
solvability of any decision problem such as CSCreate$_{poly}$ is established by the 
trivial constant-time algorithm which always answers ``Yes'' (and hence is 
always correct if a solution exists).
}
Consider
the following algorithm for {\sc Dominating set}:

\begin{enumerate}
\item Given an instance $I = \la G = (V,E), k\ra$ of {\sc Dominating set}, 
       construct an instance $I' = \la R, L_{int}, L_{comp}, $Rew()$, $Env(), 
       c$\ra$ of ESCreate$_{poly}$ using the reduction from {\sc Dominating set} to 
       CSCreate$_{poly}$ described in Lemma \ref{LemRed_DS_CSCreate1} to create
       $R$, $L_{int}$, $L_{comp}$, and $c$.
\item Run $A$ on $I'$ to produce output $O'$ for ESCreate$_{poly}$.
\item As specified in the converse part of the proof of correctness of
       the reduction in Lemma \ref{LemRed_DS_CSCreate1} use the invoked
       {\tt if-then} components in $O'$ to derive a candidate solution $O$
       for the given instance of {\sc Dominating set}.
\item If $O$ is a correct solution for $I$, output ``Yes''; otherwise,
       output ``No'' (as by the definition of promise
       solvability, if the answer was ``Yes'' then $A$ would have had
       to output $O'$ such that $O$ was a correct solution to the given
       instance of {\sc Dominating set}).
\end{enumerate}

\noindent
As all steps in this algorithm run in polynomial time, the above is
a polynomial-time algorithm for {\sc Dominating set}. However,
given the $NP$-hardness of {\sc Dominating set}, this would imply that
$P = NP$, completing the proof.
\end{proof}

\subsubsection{Polynomial-time Approximate Solvability of Restricted Emergent Software Creation}

\label{SectResESC_PTApproxSolv}

We now consider if ESCreate$_{poly}$ is efficiently approximately solvable in either of the three
senses (frequently correct (deterministic), frequently correct (probabilistic), or
approximately optimal) listed at the start of Section \ref{SectResDesSolv}.
As can be seen below, the polynomial-time exact intractability of ESCreate$_{poly}$
proved in the previous section rules out all three of these types of efficient
approximability. 

We start by considering the two types of frequently correct 
approximability.
 
\begin{description}
\item[{\bf Result A.4}] For any choice of $Rew()$ and $Env()$, if ESCreate$_{poly}$ is solvable by 
a polynomial-time algorithm with a polynomial error frequency (i.e., $err(n)$ 
is upper bounded by a polynomial of $n$) then $P = NP$.
\end{description}
\begin{proof}
That the existence of such an algorithm for CSCreate$_{poly}$ implies $P = NP$ follows 
from the $NP$-hardness of CSCreate$_{poly}$ (which is established in the proof of Result
A.1) and Corollary 2.2. in \cite{HW12}. The polynomial-time inapproximability result 
for ESCreate$_{poly}$ then follows by contradiction from Observation \ref{ObsESCreateSolvability}.
\end{proof}

\begin{description}
\item[{\bf Result A.5}] 
For any choice of $Rew()$ and $Env()$, 
if $P = BPP$ and ESCreate$_{poly}$ is polynomial-time solvable by a probabilistic 
algorithm which operates correctly with probability $\geq 2/3$ then $P = NP$.
\end{description}
%
\begin{proof}
It is widely believed that $P = BPP$ \cite[Section 5.2]{Wig07} where $BPP$ is 
considered the most inclusive class of decision problems that can be efficiently
solved using probabilistic methods (in particular, methods whose probability of
correctness is $\geq 2/3$ and can thus be efficiently boosted to be arbitrarily
close to one). Hence, if CSCreate$_{poly}$ has a probabilistic polynomial-time algorithm
which operates correctly with probability $\geq 2/3$ then CSCreate$_{poly}$ is by
definition in $BPP$. However, if $BPP = P$ and we know that CSCreate$_{poly}$
is $NP$-hard by the proof of Result A.2, this would then imply by the 
definition of $NP$-hardness that $P = NP$. The polynomial-time inapproximability result 
for ESCreate$_{poly}$ then follows by contradiction from Observation \ref{ObsESCreateSolvability}.
\end{proof}

\vspace*{0.1in}

\noindent
To assess cost-approximability, we need the following problem.

\vspace*{0.1in}

\noindent
{\sc Optimal Component-based Software Creation } (CSCreate$^{OPT}$)  \\
{\em Input}: Software system requirements $R$, interface and component libraries
              $L_{int}$ and $L_{comp}$, a base component $c \in 
              L_{comp}$, and a reward function $Rew()$. \\
{\em Output}: A working component-based software system $S$ based on
              $c$ relative to $L_{int}$, $L_{comp}$, and $R$ that 
              has the smallest value of $Rew(S)$ over all working systems
              based on $c$ relative to $L_{int}$, $L_{comp}$, and $R$, if such a system 
              exists, and special symbol $\bot$ otherwise.

\vspace*{0.1in}

\noindent
Let CSCreate$_{poly}^{OPT}$ be the version of CSCreate$^{OPT}$ such that any
component-based system $S$ runs in time polynomial in the input size $|I|$. Note that each 
input for ESCreate$_{poly}$ has a corresponding input to CSCreate$_{poly}^{OPT}$ (namely, the
input to ESCreate$_{poly}$ without $Env()$). Moreover, any algorithm $A$ that
solves ESCreate$_{poly}$ under some $Rew()$ can also be used to solve 
CSCreate$_{poly}^{OPT}$ (namely, return whatever $A$ run on the given input $x$ for 
CSCreate$_{poly}^{OPT}$ produces).  This yields the following useful observation.

\begin{observation}
For any choice of $Rew()$ and $Env()$, if there is an algorithm $A$ of solvability 
type $T$ for ESCreate$_{poly}$ than there is an algorithm $A'$ of solvability
type $T$ for CSCreate$_{poly}^{OPT}$ under $Rew()$.
\label{ObsESCreateOPTSolvability}
\end{observation}

\noindent
We first give a reduction that will be used to establish 
the cost-inapproximability of ESCreate under $Rew_{\#comp}()$.
This reduction builds on that in Lemma \ref{LemRed_DS_CSCreate1} by
further exploiting the ability of interfaces to be implemented by
multiple components to allow a set of {\tt BaseJ} components that
effectively encode all possible candidate dominating sets of size
1 to $|V|$ in $G$ (see Figure \ref{FigRed3}).

\begin{figure}[t]
\begin{center}
\includegraphics[width=3.5in]{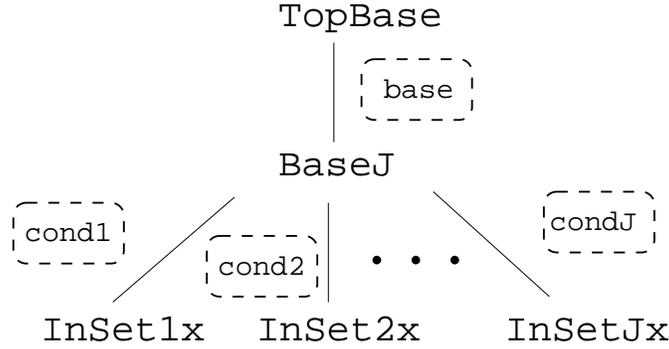}
\end{center}
\caption{General structure of valid software systems created by the reduction 
          in the proof of Lemma \ref{LemRed_DS_CSCreateOPT1}. Note that indices $J$ and $x$ 
          in {\tt InSetJx} are such that $1 \leq J, x \leq |V|$. Following the convention
          in Figure \ref{FigExESS3}, interfaces with multiple implementing
          components are enclosed in dashed boxes.
}
\label{FigRed3}
\end{figure}

\begin{lemma}
{\sc Dominating set}$^{OPT}$ polynomial-time Levin reduces to CSCreate$_{poly}^{OPT}$ under 
$Rew_{\#comp}()$ such that there is a dominating set of size $k$ for the given instance
of {\sc Dominating set}$^{OPT}$ if and only if there is a working component-based 
software system $S$ with reward value $Rew_{\#comp}(S) = k + 2$ for the constructed 
instance of CSCreate$_{poly}^{OPT}$.
\label{LemRed_DS_CSCreateOPT1}
\end{lemma}
\begin{proof}
Given an instance $I = \la G = (V,E)\ra$ of {\sc Dominating set}$^{OPT}$, construct 
the following instance $I' = \la R, L_{int}, L_{comp}, c\ra$ of CSCreate$_{poly}^{OPT}$: Let
$R$ be as in the proof of Lemma \ref{LemRed_DS_CSCreate1}.
Let $L_{int}$ consist of $|V| + 2$ interfaces broken into three groups:

\begin{enumerate}
\item{
A single interface of the form

\begin{center}
\begin{verbatim}
    interface topBase {
         void main(Input I)
    }
\end{verbatim}
\end{center}
}
\item{
A single interface of the form

\begin{center}
\begin{verbatim}
    interface base {
         void mainBase(Input I)
    }
\end{verbatim}
\end{center}
}
\item{
A set of $|V|$ interfaces of the form

\begin{center}
\begin{verbatim}
    interface condJ {
         Boolean inSetJ(Input I)
    }
\end{verbatim}
\end{center}

\noindent
for $1 \leq {\tt J} \leq |V|$.
}
\end{enumerate}

\noindent
Let $L_{comp}$ consist of $|V|^2 + |V| + 1$ components broken into three groups:

\begin{enumerate}
\item{
A single component of the form

\begin{verbatim}
    component TopBase provides topBase requires base {
        void main(Input I) {
            mainBase(I)
        }
    }
\end{verbatim}
}
\item{
A set of $|V|$ components of the form

\begin{center}
\begin{verbatim}
    component BaseJ provides base 
                   requires cond1, cond2, ..., condJ {
         void mainBase(Input I) {
             if inSet1(I) then output 1
             elsif inSet2(I) then output 1
                 ...
             elsif inSetJ(I) then output 1
             else output 0
         }
    }
\end{verbatim}
\end{center}

\noindent
for $1 \leq J \leq |V|$.
}
\item{
A set of $|V|^2$ components of the form

\begin{center}
\begin{verbatim}
    component InSetJK provides condJ {
         Boolean inSetJ(Input I) {
             return v_I(x_K) 
         }
    }
\end{verbatim}
\end{center}

\noindent
for $1 \leq {\tt J} \leq |V|$ and $1 \leq {\tt K} \leq |V|$.
}
\end{enumerate}

\noindent
Note that in $L_{comp}$, there are $|V|$ implementations of the 
{\tt base}-interface and $|V|$ implementations of each {\tt cond}-interface.  
Finally, let $c$ be component {\tt TopBase} in $L_{comp}$. Note that the 
instance of CSCreate$_{poly}^{OPT}$ described above can be constructed in 
time polynomial in the size of the given instance of 
{\sc Dominating set}$^{OPT}$.;
moreover, as there is only an at most $(|V| + 1)$-clause {\tt if-then} statement
block and
no loops in the component code and $|V| < |I|$, any candidate
component-based software system created relative to $L_{int}$, $L_{comp}$, and $c$
runs in time linear in the size of input $I'$.

Let us now verify the correctness of this reduction:

\begin{itemize}
\item{
Suppose that there is a dominating set $D$ of size $k$ in the given instance of
{\sc Dominating set}. We can then construct a component-based  software system 
consisting 
of $c$, component {\tt basek}, and the $k$ {\tt Inset}-components corresponding
to the vertices in $D$; the choice of which interface to implement for each 
vertex is immaterial. Observe that for each $(i_j,o_j) \in R$, this software 
system produces output $o_j$ given input $i_j$; moreover, $Rew_{\#comp}(S) = 
k + 2$.
}
\item{
Conversely, suppose that the constructed instance of CSCreate$_{poly}^{OPT}$ has a 
\linebreak working component-based software system based on $c$ relative to 
$L_{int}$, $L_{comp}$, and $R$ such that $Rew_{\#comp}(S) = val$.\footnote{
Note that the existence of at least one such a working system is guaranteed 
for all instances of CSCreate$_{poly}^{OPT}$ constructed as described above (namely, 
the system consisting
of components {\tt TopBase} and {\tt Base(|V|)} and the $|V|$ components
{\tt InSetJJ} for $1 \leq {\tt J} \leq | V|$, which corresponds to the
dominating set consisting of all vertices in $V$). This is necessary for
our reduction, as each instance of {\sc Dominating set}$^{OPT}$ has at
least one dominating set (namely, $V$), and cannot correspond to a 
constructed instance of CSCreate$_{poly}^{OPT}$ whose solution is $\bot$.
}
As $c$ is component {\sc TopBase} which requires a {\sc Base} component and 
this {\tt Base} component requires some number
of {\tt InSet} components, this system is comprised of components {\tt TopBase},
{\tt Base($val - 2$)}, and $val - 2$ {\tt InSet} components. 
In order to correctly 
accommodate all input-output pairs in $R$, the $(val - 2)$ {\tt if-then} 
statements in {\tt Base($val - 2$)} must implement {\tt Inset}-components whose
corresponding vertices form a dominating set in $G$ of size at most $val - 2$.
Hence, the existence of a working component-based software system $S$ such that
$Rew_{\#comp}(S) = val$ for the constructed instance of CSCreate$_{poly}^{OPT}$ implies 
the existence of a dominating set of size $val - 2$ for the given instance  of
{\sc Dominating set}.
}
\end{itemize}

\noindent
To complete the proof, note that the required functions $f()$ and $g()$ in the 
definition of a Levin reduction correspond respectively to the algorithm 
given at the beginning of this proof for constructing an instance of 
CSCreate$_{poly}^{OPT}$ under $Rew_{\#comp}()$ from the given instance of {\sc Dominating
set}$^{OPT}$ and the algorithm implicit in the converse clause of the proof of 
reduction correctness above for constructing a dominating set from a valid 
component-based software system for the constructed instance of 
CSCreate$_{poly}^{OPT}$.
\end{proof}

\vspace*{0.15in}

\noindent
The reduction above can also be used to establish the cost-inapproximability 
of ESCreate under $Rew_{CodeB}()$.

\begin{lemma}
{\sc Dominating set}$^{OPT}$ polynomial-time Levin reduces to CSCreate$_{poly}^{OPT}$ 
under $Rew_{CodeB}()$ such that there is a dominating set of size $k$ for the 
given instance of {\sc Dominating set}$^{OPT}$ if and only if there is a working
component-based software system $S$ with reward value $Rew_{CodeB}(S) = 9k + 16$ 
for the constructed instance of CSCreate$_{poly}^{OPT}$. 
\label{LemRed_DS_CSCreateOPT2}
\end{lemma}

\begin{proof}
In the proof of Lemma \ref{LemRed_DS_CSCreateOPT1}, observe that for a 
dominating set of size $k$ in the given instance of 
{\sc Dominating set}$^{OPT}$, a software system $S$ for the constructed instance
of CSCreate$_{poly}^{OPT}$ consists of components {\tt TopBase} and {\tt Basek}, $k$ 
{\tt InSet} components, interfaces {\tt topBase} and {\tt base}, and $k$ 
{\tt cond} interfaces. The total number of lines of code in this system and 
hence the value of $Rew_{CodeB}(S)$) is therefore $5 + (k + 5) + 5k + 3 + 3 + 3k 
= 9k + 16$. The result then follows by a slight modification to the proof of 
correctness of the reduction described in Lemma \ref{LemRed_DS_CSCreateOPT1}.
\end{proof}

\begin{description}
\item[{\bf Result A.6}] 
For any choice of $Env()$, if ESCreate$_{poly}^{OPT}$ under $Rew_{\#comp}()$ is polynomial-time
\newline
$c$-approximable for any constant $c > 0$ then $P = NP$.
\end{description}
\begin{proof}
Observe that in the proof of the reduction in Lemma 
\ref{LemRed_DS_CSCreateOPT1}, the size $k$ of a dominating set in $G$ in the 
given instance of {\sc Dominating set}$^{OPT}$ is always a linear function of 
the value of $Rew_{\#comp}(S)$ in the constructed instance of CSCreate$_{poly}^{OPT}$, 
i.e., $k = Rew_{\#comp}(S) - 2$.  This means that a polynomial-time 
$c$-approximation algorithm for CSCreate$_{poly}^{OPT}$ under $Rew_{\#comp}()$ for any 
constant $c$, when combined with the reduction from {\sc Dominating set}$^{OPT}$
to CSCreate$_{poly}^{OPT}$ described in the proof of Lemma 
\ref{LemRed_DS_CSCreateOPT1}, implies the existence of a polynomial-time 
$3c$-approximation algorithm for {\sc Dominating set}${}^{OPT}$ (as $c \times 
Rew_{\#comp}(S) = c \times (k + 2) \leq c \times (k + 2k) \leq 3c \times k$ for 
$k \geq 1$). However, if {\sc Dominating set}${}^{OPT}$ has a polynomial-time 
$c$-approximation algorithm for any constant $c > 0$ then $P = NP$ \cite{LY94},
which means that CSCreate$_{poly}^{OPT}$ under $Rew_{\#comp}()$ cannot have a 
polynomial-time $c$-approximation algorithm for any $c > 0$ unless $P = NP$. 
The polynomial-time inapproximability result for ESCreate$_{poly}^{OPT}$ under
$Rew_{\#comp}()$ then follows by 
contradiction from Observation \ref{ObsESCreateOPTSolvability}.
\end{proof}

\begin{description}
\item[{\bf Result A.7}] 
For any choice of $Env()$, if ESCreate$_{poly}^{OPT}$ under $Rew_{CodeB}()$ is polynomial-time 
\newline
$c$-approximable for any constant $c > 0$
then $P = NP$.
\end{description}
\begin{proof}
Observe that in the 
proof of the reduction in Lemma \ref{LemRed_DS_CSCreateOPT2}, the size $k$ of 
a dominating set in $G$ in the given instance of {\sc Dominating set}$^{OPT}$ 
is always a linear function of the value of $Rew_{CodeB}(S)$ in the constructed 
instance of CSCreate$_{poly}^{OPT}$, i.e., $k = \frac{Rew_{CodeB}(S) - 16}{9}$.  This 
means that a polynomial-time $c$-approximation algorithm for CSCreate$_{poly}^{OPT}$ 
under $Rew_{CodeB}()$ for any constant $c$, when combined with the reduction from
{\sc Dominating set}$^{OPT}$ to CSCreate $^{OPT}$ described in the proof of 
Lemma \ref{LemRed_DS_CSCreateOPT2}, implies the existence of a polynomial-time 
$25c$-approximation algorithm for {\sc Dominating set}${}^{OPT}$ (as $c \times 
Rew_{CodeB}(S) = c \times (9k + 16) \leq c \times (9k + 16k) \leq 25c \times k$ 
for $k \geq 1$). However, if {\sc Dominating set}${}^{OPT}$ has a 
polynomial-time $c$-approximation algorithm for any constant $c > 0$ then 
$P = NP$ \cite{LY94}, which means that CSCreate$_{poly}^{OPT}$ under $Rew_{CodeB}()$ 
cannot have a polynomial-time $c$-approximation algorithm for any $c > 0$ 
unless $P = NP$. The polynomial-time inapproximability result for ESCreate$_{poly}^{OPT}$ 
under $Rew_{CodeB}()$ then follows by contradiction from Observation 
\ref{ObsESCreateOPTSolvability}.
\end{proof}

\subsubsection{Fixed-parameter Tractability of Restricted Emergent Software Creation}

\label{SectResESC_FPSolv}

\begin{table}[t]
\caption{
Parameters for emergent software creation problems.
}
\label{TabPrm}
\centering
\begin{tabular}{| c ||  l | }
\hline
Parameter   & Description \\
\hline\hline
$|L_{int}|$ & \# available interfaces \\
\hline
$|L_{comp}|$ & \# available components \\
\hline\hline
$I_{ci}$    & Maximum \# components implementing an interface \\
\hline\hline
$C_{pi}$    & Maximum \# provided interfaces per component \\
\hline
$C_{ri}$    & Maximum \# required interfaces per component \\
\hline\hline
$S_{comp}$  & Maximum \# components in a valid system \\
\hline
$S_{depth}$ & Maximum depth of component wiring tree \\
\hline
\end{tabular} 
\end{table}

Given the plethora of intractability results in the previous three
subsections, we now consider to what extent and relative to which
parameters ESCreate$_{poly}$ is and is not fp-tractable.
In our analyses below, we will focus on parameter-sets $K$ drawn from 
the parameters listed in Table \ref{TabPrm}. These parameters can be divided
into four main groups:

\begin{enumerate}
\item Parameters characterizing interface and component libraries 
       ($|L_{int}|, |L_{comp}|$);
\item Parameters charactering interfaces 
       ($I_{ci}$);
\item Parameters charactering components 
       ($C_{pi}, C_{ri}$); and
\item Parameters characterizing component-based software systems 
       ($S_{comp}, S_{depth}$).
\end{enumerate}

\noindent
We first consider those parameter-sets which yield fp-intractability.

\begin{description}
\item[{\bf Result A.8}] For any choice of $Rew()$ and $Env()$, if $\langle |L_{int}|,
                   C_{pi}, C_{ri}, S_{comp}, S_{depth}\ra$-ESCreate$_{poly}$ is fp-tractable
                   then $FPT = W[2]$.
\end{description}
\begin{proof}
Given the $W[2]$-hardness of $\la k \ra$-{\sc Dominating set}, the reduction in
Lemma \ref{LemRed_DS_CSCreate1} implies that CSCreate$_{poly}$ is $W[2]$-hard when
$C_{pi} = 1$, $S_{depth} = 2$, $C_{ri} = k$, and $|L_{int}| = S_{comp} = k + 1$ and 
hence not fp-tractable relative to these parameters unless $FPT = W[2]$.
The fp-intractability result for ESCreate$_{poly}$ then follows by contradiction 
from Observation \ref{ObsESCreateSolvability}.
\end{proof}

\vspace*{0.15in}

\noindent
The reductions underlying the following three results exploit the tricks 
previously used to such good effect in Lemmas \ref{LemRed_DS_CSCreate1} and 
\ref{LemRed_DS_CSCreateOPT1} as well as other features of our software component model.
The reduction underlying Result A.9 reduces the number of {\tt InSet} components by
invoking a larger encoding of candidate dominating sets and more complex but still
polynomial-time checking computations in the {\tt Base} component.
The reduction underlying Result A.10 reduces the number of interfaces required by
any component to a constant by splitting the creation of the candidate dominating sets
in component {\tt Base} in the reduction in the proof of Result A.9 over multiple
components. Finally, the reduction underlying Result A.11 reduces the number of
components in $L_{comp}$ to 3 by exploiting the ability of components providing
multiple interfaces to provide only the code required by an interface in that
interface's copy of the component.
Readers interested in details can consult the full proofs of these
results in the appendix.

\begin{description}
\item[{\bf Result A.9}] For any choice of $Rew()$ and $Env()$, if $\langle I_{ci}, C_{pi}, 
                   S_{depth}\ra$-ESCreate$_{poly}$ is fp-tractable then 
                   $P = NP$.
\end{description}

\begin{description}
\item[{\bf Result A.10}] For any choice of $Rew()$ and $Env()$, if $\langle I_{ci}, C_{pi}, 
                    C_{ri}\ra$-ESCreate$_{poly}$ is fp-tractable then $P = NP$.
\end{description}

\begin{description}
\item[{\bf Result A.11}] For any choice of $Rew()$ and $Env()$, if  $\langle |L_{comp}|, 
                    I_{ci}, S_{depth}\ra$-ESCreate$_{poly}$ is fp-tractable then 
                    $P = NP$.
\end{description}

\noindent
We now consider those parameter-sets that yield fp-tractability. All
of these results are based on the same brute-force solution enumeration
algorithm relative to different worst-case runtime analyses.

\begin{description}
\item[{\bf Result A.12}] For any choice of $Rew()$ and $Env()$, $\langle I_{ci}, C_{ri}, 
                    S_{depth}\ra$-ESCreate$_{poly}$ is fp-tractable.
\end{description}
\begin{proof}
The largest possible component-based software system relative to a given 
$L_{int}$ and $L_{comp}$ has a component wiring tree 
rooted at base component $c$ with branching factor $C_{ri}$ and depth 
$S_{depth}$. This tree has $(C_{ri})^{S_{depth}} - 2$ non-root vertices, each
corresponding to an interface required by a component. As each of these 
interfaces can be implemented by  at most $I_{ci}$ components, there are at most
$(I_{ci} + 1)^{(C_{ri})^{S_{depth}}}$ possible component-based software systems 
of depth at most $S_{depth}$ based on $c$ (the ``+ 1'' term at the lowest level 
denotes labeling a vertex $v$ with a special symbol that triggers deletion all 
descendent-vertices of $v$). 

Consider the algorithm that exhaustively generates all such 
systems and for each system $S$, (i) determines if $S$ is a working system 
relative to $R$ and, if so, (2) computes reward value $Rew(S)$. The output of this
algorithm is the working system with the lowest or highest reward value,
depending on the intent of $Rew()$. Given the above and our assumption that a 
candidate component-based software system $S$ can be checked against software 
system requirements $R$ in time polynomial in the sizes of $S$ and $R$,
this algorithm runs in fp-time relative to $I_{ci}$, $C_{ri}$, and $S_{depth}$,
completing the proof of this result.
\end{proof}

\begin{description}
\item[{\bf Result A.13}] For any choice of $Rew()$ and $Env()$, $\langle I_{ci}, 
S_{comp}\ra$-ESCreate$_{poly}$ is \linebreak fp-tractable.
\end{description}
\begin{proof}
As no path from the root to a leaf in the wiring component trees for our 
software systems can contain duplicate component vertex-labels,
the length of the longest path in such a tree from base component $c$ is 
bounded by $S_{comp}$; this means that $S_{depth} \leq S_{comp}$. Moreover,
as implementing each required interface adds a component to the software system,
$C_{ri} \leq S_{comp} -1 < S_{comp}$. Given these two observations, this 
result then follows from the algorithm in the proof of Result A.12.
\end{proof}

\begin{description}
\item[{\bf Result A.14}] For any choice of $Rew()$ and $Env()$, $\langle |L_{int}|, 
|L_{comp}|\ra$-ESCreate$_{poly}$ is fp-tractable.
\end{description}
\begin{proof}
As no path from the root to a leaf in the wiring component trees for our software 
systems can contain duplicate component vertex-labels,
the length of the longest path in such a tree from base component $c$ is 
bounded by $|L_{comp}|$; this means that $S_{depth} \leq |L_{comp}|$. Moreover,
as a component cannot require the same interface twice, $C_{ri} \leq |L_{int}|$.
Given these two observations, this 
result then follows from the algorithm in the proof of Result A.12.
\end{proof}

\vspace*{0.15in}

\noindent
Note that for each of the parameter-sets in Results A.12--A.14, ESCreate$_{poly}$
is fp-intractable relative to each non-empty subset of these
parameter-sets. Hence, these fp-tractability results are all minimal,
in the sense that no subsets of the parameters in their associated
parameter-sets yield fp-tractability.

\subsection{Results for Emergent Software Adaptation}

\label{SectResESA}

Many of the results derived in this section for ESAdapt will actually be 
derived relative to the following problem:

\vspace*{0.1in}

\noindent
{\sc Component-based Software Adaptation } (CSAdapt)  \\
{\em Input}: Software system requirements $R$, interface and component libraries
              $L_{int}$ and $L_{comp}$, a working component-based software 
              system $S$ based on component $c \in L_{comp}$ relative to $R$, 
              $L_{int}$, and $L_{comp}$, reward function $Rew()$, and an integer $k$. \\
{\em Question}: Is there a working component-based software system $S'$ based on $c$ 
              relative to $L_{int}$, $L_{comp}$, and $R$ such that $Rew(S') \leq k$?

\vspace*{0.1in}

\noindent
Note that each input for ESAdapt has a corresponding input to CSAdapt (namely, the input 
to ESAdapt without $Env()$). Moreover, any algorithm $A$ that solves 
ESAdapt under some $Rew()$ can be also used to solve CSAdapt (namely, if $A$ run on the 
given input $x$ for CSAdapt produces a working system $S'$ such that $Rew(S') \leq k$, 
output ``Yes'', otherwise output ``No''). This yields the following useful observation.

\begin{observation}
For any choice of $Rew()$ and $Env()$, if there is an algorithm $A$ of solvability 
type $T$ for ESAdapt under $Rew()$ than there is an algorithm $A'$ of solvability
type $T$ for CSAdapt under $Rew()$.
\label{ObsESAdaptSolvability}
\end{observation}

\noindent
It is very important to note that CSAdapt (unlike CSCreate in Section \ref{SectResESC}) 
explicitly invokes $Rew()$; hence, all of our results proved by invoking Observation
\ref{ObsESAdaptSolvability} (i.e., all results proved in this section) are relative to 
specific $Rew()$, namely, either $Rew_{\#comp}()$ or $Rew_{CodeB}()$. This has some 
interesting consequences, which will be discussed further in Section \ref{SectDisc}.

\subsubsection{Unsolvability of Unrestricted Emergent Software Adaptation}

\label{SectResESAUnsolv}

We start off by considering if problem ESAdapt is solvable in the most
general possible case --- that is, if ESAdapt has an algorithm that always
returns the correct output for an input incorporating any possible choice of
$Env()$ and in which there are no restrictions on the form, size,
or running times of $L_{int}$ and $L_{comp}$, their member interfaces and
components, or any software systems created using $L_{int}$ and $L_{comp}$. 
Analogous to ESCreate in Section \ref{SectResESCUnsolv}, this once again 
turns out not to be the case, though we only show
unsolvability at this time for ESAdapt under either $Rew_{\#comp}()$
or $Rew_{CodeB}()$.

Let us first consider ESAdapt under $Rew_{\#comp}()$.
In the following, we modify $L_{int}$ and $L_{comp}$ in the reduction in the 
proof of Result A.1 such that any working component-based software system has
a new topmost component {\tt topBase} and there are thus now two possible
working component-based software systems, one of which is the given ``twin''
component-based system $S$ with extra component ({\tt Base1a}) that can be 
disallowed as a possible solution by appropriately setting the value of $k$ 
relative to reward function $Rew_{\#comp}()$ 
(see Figure \ref{FigRed6}).

\begin{figure}[t]
\begin{center}
\includegraphics[width=3.5in]{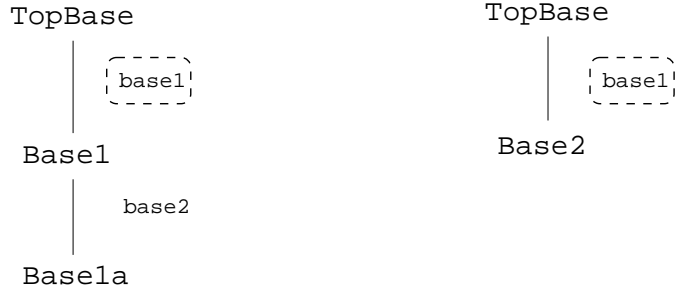}
\end{center}
\caption{General structure of valid software systems created by the reduction 
          in the proof of Result B.1. The ``twin'' component-based system with an extra
          component and extra code that is given as $S$ in the reduction is on the left,
          and is forbidden as output by appropriate values of $k$ under $Rew_{\#comp}()$
          and $Rew_{CodeB}()$ in favor of the component-based system on the right.
          Following the convention
          in Figure \ref{FigExESS3}, interfaces with multiple implementing
          components are enclosed in dashed boxes.
}
\label{FigRed6}
\end{figure}

\begin{description}
\item[{\bf Result B.1}] For any choice of $Env()$, ESAdapt under $Rew_{\#comp}()$ is 
unsolvable.
\end{description}
\begin{proof}
Consider the following polynomial-time Karp reduction from {\sc TM Halting}
to CSAdapt: given an instance $I = \la M, x\ra$ of {\sc TM Halting}, construct 
an instance $I' = \la R, L_{int}, L_{comp}, S, c, Rew() = Rew_{\#comp}(),
k\ra$ of CSAdapt in which $X = \{x_1\}$ and $O = \{1\}$, there
is a single input-output pair $r$ in $R$ such that for $r = (True, 1)$,
$L_{int}$ consists of the three interfaces

\begin{center}
\begin{verbatim}
    interface topBase {
         void main(Input I)
    }

    interface base1 {
         void main1(Input I)
    }

    interface base2 {
         void base1a(Input I)
    }
\end{verbatim}
\end{center}

\noindent
and $L_{comp}$ consists of the four components

\begin{center}
\begin{verbatim}
    component TopBase provides topBase requires base1 {
         void main(Input I) {
             main1(I)
         }
    }

    component Base1 provides base1 requires base2 {
         void main1(Input I) {
             base1a(I)
         }
    }

    component Base1a provides base2 {
         void base1a(Input I) {
             output 1
         }
    }

    component Base2 provides base1 {
         void main1(Input I) {
             <CODEM(x)>
             output 1
         }
    }
\end{verbatim}
\end{center}

\noindent 
where {\tt <CODEM(x)>} is the Dana code simulating the computation of
$M$ on input $x$. As noted previously in the proof of Result A.1,
Dana contains both loops and conditional statements and can readily 
simulate $M$ on input $x$ using code that is of size polynomial
in the sizes of the given descriptions of $M$ and $x$. Finally, let $c$
be component {\tt TopBase} in $L_{comp}$, $S$ be the software system
based on components {\tt TopBase}, {\tt Base1} and {\tt Base1a}, and $k = 2$. 
Note that
the instance of CSAdapt described above can be constructed in time 
polynomial in the size of the given instance of {\sc TM Halting}. To 
conclude the proof, observe that the only possible component-based 
system for the constructed instance of CSAdapt which has only the
two components required by the value of $k$ is that consisting of 
{\tt TopBase} and {\tt Base2}, and that this system satisfies the sole 
input-output constraint in $R$ if and only if $M$ halts on input $x$ for the 
given instance of {\sc TM Halting}.
It is known that {\sc TM Halting} cannot have an algorithm that is correct 
for all given $\langle M, x\rangle$ instances (\cite{Tur36}; see also 
\cite[Section 9.2.4]{HMU01}), and 
hence is unsolvable. The reduction above implies in
turn that CSAdapt under $Rew_{\#comp}()$ cannot have an algorithm either. 
The unsolvability result for ESAdapt under $Rew_{\#comp}()$ then 
follows by contradiction from Observation \ref{ObsESAdaptSolvability}.
\end{proof}

\noindent
We can in turn use a version of the reduction above to prove unsolvability of ESAdapt 
under $Rew_{CodeB}()$ by ``padding'' the code of component {\tt Base1a} to ensure that the 
given component-based system $S$ cannot be a solution relative to an appropriate
value of $k$.

\begin{description}
\item[{\bf Result B.2}] For any choice of $Env()$, ESAdapt under $Rew_{CodeB}()$ is unsolvable.
\end{description}
\begin{proof}
Consider a modification of the reduction in the proof of Result B.1
in which component {Base1a} is instead

\begin{center}
\begin{verbatim}
    component Base1a provides base2 {
         void base1a(Input I) {
             <BLOCK>
             output 1
         }
    }
\end{verbatim}
\end{center}

\noindent
where {\tt <BLOCK>} consists of $|${\tt <CODE(M(x))>}$|$ copies
of the statement {\tt x = 1} and $k =|${\tt <CODE(M(x))>}$| + 16$. To
conclude the proof, observe that the only possible component-based 
system for the constructed instance of CSAdapt which has the 
code-base length required by the value of $k$ is that consisting of 
{\tt TopBase} and {\tt Base2}, and that this system satisfies the sole 
input-output constraint in $R$ if and only if $M$ halts on input $x$ for the 
given instance of {\sc TM Halting}. It is known that {\sc TM Halting} cannot 
have an algorithm that is correct for all given $\langle M, x\rangle$ instances
\cite[Section 9.2.4]{HMU01}, and hence is unsolvable. The reduction above 
implies in turn that CSAdapt under $Rew_{CodeB}()$ cannot have an algorithm 
either. The unsolvability result for ESAdapt under $Rew)_{CodeB}()$ 
then follows by contradiction from Observation \ref{ObsESAdaptSolvability}.
\end{proof}

\vspace*{0.15in}

\noindent
Unlike ESCreate in Section \ref{SectResESAUnsolv}, the two results above only hold relative
to {\em some} choices of $Env{}$ and $Rew()$ --- namely, every possible choice of $Env()$
paired with either of $Rew_{\#comp}()$ or $Rew_{CodeB}()$  (this is because  the proofs of 
these results only ignore $Env()$ and not both $Env()$ and $Rew()$).
That being said, restricted versions of 
ESAdapt under $Rew_{\#comp}()$ and $Rew_{CodeB}()$ may yet have correct and
even efficient algorithms. Thus, in the remainder of this paper, we shall 
assume that ESAdapt and CSAdapt are (as were ESCreate and CSCreate
in Section \ref{SectResESCUnsolv})
restricted such that any component-based software system $S$ and $S'$ 
relative to given $L_{int}$ and $L_{comp}$ run in time polynomial in 
input size  $|I|$ and hence can be checked against the system requirements in 
$R$ in time polynomial in the sizes of $S$ and $R$, i.e., created 
software systems not only operate but can also be verified quickly. We will 
denote these restricted versions as ESAdapt$_{poly}$ and CSAdapt$_{poly}$, respectively.
 
\subsubsection{Polynomial-time Exact Solvability of Restricted Emergent Software Adaptation}

\label{SectResESA_PTExactSolv}

We now consider if ESAdapt$_{poly}$ is efficiently solvable in the first of the senses listed at the
start of Section \ref{SectResDesSolv} --- namely, polynomial-time exact solvability and
polynomial-time exact promise solvability.
As was the case with ESCreate$_{poly}$ in Section \ref{SectResESC_PTExactSolv},
it turns out that ESAdapt$_{poly}$ is also polynomial-time exact intractable,
with the difference
that this intractability is proven relative to specific rather than any
reward functions (namely, $Rew_{\#comp}()$ and $Rew_{CodeB}()$).

\begin{lemma}
{\sc Dominating set} polynomial-time Karp reduces to CSAdapt$_{poly}$ under 
$Rew_{\#comp}()$.
\label{LemRed_DS_CSAdapt1}
\end{lemma}

\begin{proof}
Given an instance $I = \la G = (V,E), k\ra$ of {\sc Dominating set}, construct 
the following instance $I' = \la R, L_{int}, L_{comp}, S, c, Rew() = Rew_{\#comp}(),
k'\ra$ of CSAdapt$_{poly}$: 
Without loss of generality, assume $k < |V|$. Let $R$, $L_{int}$, $L_{comp}$, 
and $c$ be as in the proof of Lemma \ref{LemRed_DS_CSCreateOPT1} and $S$ be the
software system based on {\tt TopBase}, {\tt Base(|V|)}, and the components 
{\tt InSetJJ} for $1 \leq {\tt J} \leq |V|$. Finally, let $k' = k + 2$. Note 
that this instance CSAdapt$_{poly}$ can be constructed in time polynomial in the size 
of the given instance of {\sc Dominating set}; moreover, as there are only a 
$(k + 1) \leq (|V| + 1) < (|I| + 1)$-clause {\tt if-then} statement block and 
two single-level loops in the component code that each execute at most $|V| < 
|I|$ times, any component-based software system created relative to $L_{int}$, 
$L_{comp}$, and $Ec$ runs in time linear in the size of input $I'$.

Let us now verify the correctness of this reduction:

\begin{itemize}
\item{
Suppose that there is a dominating set $D$ of size $k$ in the given instance of 
{\sc Dominating set}. We c an then construct a component-based software system $S'$ 
consisting of $c$ and component {\sc Basek}, in which the $k$ 
{\tt cond}-interfaces in {\tt Basek} are implemented by {\tt Inset} components
corresponding to the vertices in $D$.  Observe that for each $(i_j, o_j) \in 
R$, $S'$ produces output $o_j$ given input $i_j$; moreover, $Rew_{\#comp}(S') = 
k + 2 \leq k'$.
}
\item{
Conversely, suppose that the constructed instance of CSAdapt$_{poly}$ has a \linebreak 
working component-based software system $S'$ based on $c$ relative to $L_{int}$,
$L_{comp}$, and $R$ such that $Rew_{\#comp}(S') \leq k' = k + 2$. Such a system 
$S'$ cannot include component {\tt Base(|V|)} as in $S$, because in that case,
$Rew_{\#comp}(S')$ would have the value $|V| + 2 \not\leq k'$. Therefore, $S'$ 
must include one of the component {\tt BaseJ} for $1 \leq {\tt J} \leq k$; let 
us call this component {\tt BaseM}. As in the proof of Lemma 
\ref{LemRed_DS_CSCreateOPT1}, 
in order to correctly accommodate all
other input-output pairs in $R$, the $M$ {\tt if-then} statements in 
{\sc BaseM} must implement {\tt InSet}-components whose corresponding vertices 
form a dominating set in $G$ of size at most $k$ in $G$. Hence, the existence of
a working component-based software system for the constructed instance of 
CSCreate$_{poly}$ implies the existence of a dominating set of size at most $k$ for the 
given instance  of {\sc Dominating set}.
}
\end{itemize}

\noindent
This completes the proof.
\end{proof}

\begin{lemma}
{\sc Dominating set} polynomial-time Karp reduces to CSAdapt$_{poly}$ under 
$Rew_{CodeB}()$.
\label{LemRed_DS_CSAdapt2}
\end{lemma}

\begin{proof}
In the proof of Lemma \ref{LemRed_DS_CSAdapt1}, observe that for a dominating 
set of size $k$ in the given instance of {\sc Dominating set}$^{OPT}$, a 
software system $S'$ for the constructed instance of CSAdapt$_{poly}$ consists of 
components {\tt TopBase} and {\tt Basek}, $k$ {\tt InSet} components, 
interfaces {\tt topBase} and {\tt base}, and $k$ {\tt cond} interfaces. The 
total number of lines of code in this system and hence the value of 
$Rew_{CodeB}(S')$) is therefore $5 + (k + 5) + 5k + 3 + 3 + 3k = 9k + 16$. The 
result then follows by a slight modification to the proof of correctness of 
the reduction described in Lemma \ref{LemRed_DS_CSAdapt1}.
\end{proof}

\begin{description}
\item[{\bf Result B.3}] For any choice of $Env()$, if ESAdapt$_{poly}$ under $Rew_{\#comp}()$
is polynomial-time exact solvable then $P = NP$.
\end{description}
\begin{proof}
Given the $NP$-hardness of {\sc Dominating set}, the reduction in Lemma
\ref{LemRed_DS_CSAdapt1} implies that CSAdapt$_{poly}$ under $Rew_{\#comp}()$ is $NP$-hard,
and hence not solvable in polynomial time unless $P = NP$. The polynomial-time
intractability result for ESAdapt$_{poly}$ under $Rew_{\#comp}()$ then follows by 
contradiction from Observation \ref{ObsESAdaptSolvability}.
\end{proof}

\begin{description}
\item[{\bf Result B.4}] For any choice of $Env()$, if ESAdapt$_{poly}$ under $Rew_{CodeB}()$ is
polynomial-time exact solvable then $P = NP$.
\end{description}
\begin{proof}
Given the $NP$-hardness of {\sc Dominating set}, the reduction in Lemma
\ref{LemRed_DS_CSAdapt2} implies that CSAdapt$_{poly}$ under $Rew_{CodeB}()$is $NP$-hard, 
and hence not solvable in polynomial time unless $P = NP$. The polynomial-time
intractability result for ESAdapt$_{poly}$ under $Rew_{CodeB}()$ then follows by 
contradiction from Observation \ref{ObsESAdaptSolvability}.
\end{proof}

\vspace*{0.15in}

\noindent
Note that, unlike for ESCreate$_{poly}$, evaluating the polynomial-time exact promise
solvability of ESAdapt$_{poly}$ is not possible. This is because (as noted 
previously in Section \ref{SectForm}) any version of ESAdapt$_{poly}$ promising a working 
software system for the given input will always have at least one such system --- namely, $S$.

\subsubsection{Polynomial-time Approximate Solvability of Restricted Emergent Software Adaptation}

\label{SectResESA_PTApproxSolv}

We now consider if ESAdapt$_{poly}$ is efficiently approximately solvable in either of the three
senses (frequently correct (deterministic), frequently correct (probabilistic), or
approximately optimal) listed at the start of Section \ref{SectResDesSolv}.
Once again, as with ESAdapt$_{poly}$, this turns out not to be the case.
The proofs of Results B.5--B.8 below are analogous to the proofs for Results 
A.2 and A.3 in Section \ref{SectResESC_PTExactSolv}, only this time relative to
the proofs of Results B.3 and B.4 and Observation \ref{ObsESAdaptSolvability}.

\begin{description}
\item[{\bf Result B.5}] For any choice of $Env()$, if ESAdapt$_{poly}$ under $Rew_{\#comp}()$
is solvable by a polynomial-time algorithm with a polynomial error frequency 
(i.e., $err(n)$ is upper bounded by a polynomial of $n$) then $P = NP$.
\end{description}

\begin{description}
\item[{\bf Result B.6}] For any choice of $Env()$, if ESAdapt$_{poly}$ under $Rew_{CodeB}()$
is solvable by a polynomial-time algorithm with a polynomial error frequency 
(i.e., $err(n)$ is upper bounded by a polynomial of $n$) then $P = NP$.
\end{description}

\begin{description}
\item[{\bf Result B.7}] For any choice of $Env()$, if $P = BPP$ and ESAdapt$_{poly}$ under 
$Rew_{\#comp}()$ is polynomial-time solvable by a probabilistic 
algorithm which operates correctly with probability $\geq 2/3$ then $P = NP$.
\end{description}

\begin{description}
\item[{\bf Result B.8}] For any choice of $Env()$, if $P = BPP$ and ESAdapt$_{poly}$ under
$Rew_{CodeB}()$ is polynomial-time solvable by a probabilistic 
algorithm which operates correctly with probability $\geq 2/3$ then $P = NP$.
\end{description}

\noindent
To assess cost-approximability, we need the following problem.

\vspace*{0.1in}

\noindent
{\sc Optimal Component-based Software Adaptation } (CSAdapt$^{OPT}$)  \\
{\em Input}: Software system requirements $R$, interface and component libraries
              $L_{int}$ and $L_{comp}$, a working component-based software
              system $S$ based on component $c \in L_{comp}$ relative to $R$,
              $L_{int}$, and $L_{comp}$, and a reward function $Rew()$. \\
{\em Output}: A working component-based software system $S'$ based on
               $c$ relative to $L_{int}$, $L_{comp}$, and $R$ that has the
               smallest possible value of $Rew(S)$ over all working systems based on $c$
               relative to $L_{int}$, $L_{comp}$, and $R$.

\vspace*{0.1in}

\noindent
Let CSAdapt$_{poly}^{OPT}$ be the version of CSAdapt$^{OPT}$ such that given
and candidate component-based systems $S$ and $S'$ run in time polynomial in the input 
size $|I|$. Note that each input for ESAdapt$_{poly}$ has a corresponding input to 
CSAdapt$_{poly}^{OPT}$ (namely, the input to ESAdapt$_{poly}$ without $Env()$). 
Moreover, any algorithm $A$ that solves ESAdapt$_{poly}$ under some $Rew()$ 
can also be used to solve CSAdapt$_{poly}^{OPT}$ (namely, return whatever $A$ run on the 
given input $x$ for CSAdapt$_{poly}^{OPT}$ produces). This yields the following useful 
observation.

\begin{observation}
For any choice of $Rew()$ and $Env()$, if there is an algorithm $A$ of solvability type
$T$ for ESAdapt$_{poly}$ under $Rew()$ than there is an algorithm $A'$ of solvability type
$T$ for CSAdapt$_{poly}^{OPT}$ under $Rew()$.
\label{ObsESAdaptOPTSolvability}
\end{observation}

\noindent
Observe that the reductions in the proof of Lemmas \ref{LemRed_DS_CSAdapt1} and
\ref{LemRed_DS_CSAdapt2} based on that in the proof of Lemma 
\ref{LemRed_DS_CSCreateOPT1} only add a working software component-based 
software system $S$ and do not change any other details of the constructed $R$,
$L_{int}$, $L_{comp}$, and $c$. Hence, the following hold by slight 
modifications to the proofs of Lemmas \ref{LemRed_DS_CSCreateOPT1} and
\ref{LemRed_DS_CSCreateOPT2} and Results A.6 and A.7
 
\begin{lemma}
{\sc Dominating set}$^{OPT}$ polynomial-time Levin reduces to CSAdapt$_{poly}^{OPT}$ 
under $Rew_{\#comp}()$ such that there is a dominating set of size $k$ for the 
given instance of {\sc Dominating set}$^{OPT}$ if and only if there is a 
working component-based software system $S$ with reward value 
$Rew_{\#comp}(S) = k + 2$ for the constructed instance of CSAdapt$_{poly}^{OPT}$.
\label{LemRed_DS_CSAdaptOPT1}
\end{lemma}

\begin{lemma}
{\sc Dominating set}$^{OPT}$ polynomial-time Levin reduces to CSAdapt$_{poly}^{OPT}$ 
under $Rew_{CodeB}()$ such that there is a dominating set of size $k$ for the 
given instance of {\sc Dominating set}$^{OPT}$ if and only if there is a working
component-based software system $S$ with reward value $Rew_{CodeB}(S) = 9k + 16$ 
for the constructed instance of CSAdapt$_{poly}^{OPT}$. 
\label{LemRed_DS_CSAdaptOPT2}
\end{lemma}

\begin{description}
\item[{\bf Result B.9}] 
For any choice of $Env()$, if  ESAdapt$^{OPT}$ under $Rew_{\#comp}()$ is polynomial-time
\newline
$c$-approximable for any constant $c > 0$ then $P = NP$.
\end{description}

\begin{description}
\item[{\bf Result B.10}] 
For any choice of $Env()$, if ESAdapt$^{OPT}$ under $Rew_{CodeB}()$ is polynomial-time 
\newline
$c$-approximable for any constant $c > 0$ then $P = NP$.
\end{description}

\subsubsection{Fixed-parameter tractability of Restricted Emergent Software Adaptation}

\label{SectResESA_FPSolv}

Given the plethora of intractability results in the previous three
subsections, we now consider to what extent and relative to which
parameters ESAdapt$_{poly}$ is and is not fp-tractable.
Our results in this section are derived relative to the parameters listed
in Table \ref{TabPrm} for ESCreate$_{poly}$ in Section \ref{SectResESC_FPSolv} and use 
versions of the proofs of Results A.8, A.9, A.10, and A.11 as modified by the 
``twinning'' and ``padding'' tricks used in the proofs of Results B.1 and B.2. 
The former trick works relative to the results listed above derived by 
reductions from
{\sc Dominating set}, as any instance of {\sc Dominating set} always has a 
trivial dominating set consisting of the set $V$ of all vertices in the given 
graph $G$ which can be used to structure the ``twin'' given component-based 
software system $S$.  We show how this is done relative to the proof of 
Result A.8 for ESAdapt$_{poly}$ under $Rew_{\#comp}()$ and $Rew_{CodeB}()$ and 
leave the details of the other proofs to the reader.

\begin{description}
\item[{\bf Result B.11}] For any choice of $Env()$, if $\langle |L_{int}|, C_{pi}, C_{ri}, 
                   S_{comp}, S_{depth}\ra$-ESAdapt$_{poly}$ under $Rew_{\#comp}()$ is 
                   fp-tractable then $FPT = W[2]$.
\end{description}
\begin{proof}
Consider the following polynomial-time Karp reduction from 
{\sc Dominating set} to CSAdapt$_{poly}$ under $Rew_{\#comp}()$: given an instance
$I = \la G = (V,E), k\ra$ of {\sc Dominating set}, construct an 
instance $I' = \la R, L_{int}, L_{comp}, S, c, Rew() =$ \linebreak $Rew_{\#comp}(), k'\ra$ of
CSAdapt$_{poly}$ in which $X = \{x_1, x_2, \ldots, x_{|V|}\}$, i.e., there is a 
unique Boolean variable corresponding to each vertex in $V$, and $O = 
\{0, 1\}$. There are $|V|$ input-output pairs in $R$ such that for $r_i
= (i_j, o_j)$, $1 \leq j \leq |V|$, $v_{i_j}(x_k) = True$ if $v_k \in 
N_C(v_j)$ and is $False$ otherwise and $o_j = 1$.
Let $L_{int}$ consist of $k + 3$ interfaces broken into four groups:

\begin{enumerate}
\item{
A single interface of the form

\begin{center}
\begin{verbatim}
    interface topBase {
         void main(Input I)
    }
\end{verbatim}
\end{center}
}
\item{
A single interface of the form

\begin{center}
\begin{verbatim}
    interface base1 {
         void main1(Input I)
    }
\end{verbatim}
\end{center}
}
\item{
A single interface of the form

\begin{center}
\begin{verbatim}
    interface base2 {
         void base1a(Input I)
    }
\end{verbatim}
\end{center}
}
\item{
A set of $k$ interfaces of the form

\begin{center}
\begin{verbatim}
    interface condJ {
         Boolean inSetJ(Input I)
    }
\end{verbatim}
\end{center}

\noindent
for $1 \leq {\tt J} \leq k$.
}
\end{enumerate}

\noindent
Let $L_{comp}$ consist of $k|V| + 4$ components broken into five groups:

\begin{enumerate}
\item{
A single component of the form

\begin{center}
\begin{verbatim}
    component TopBase provides topBase requires base1 {
         void main(Input I) {
             main1(I)
         }
    }
\end{verbatim}
\end{center}
}
\item{
A single component of the form

\begin{center}
\begin{verbatim}
    component Base1 provides base1 requires base2 {
         void main1(Input I) {
             base1a(I)
         }
    }
\end{verbatim}
\end{center}
}
\item{
A single component of the form

\begin{center}
\begin{verbatim}
    component Base1a provides base2 {
         void base1a(Input I) 
                   requires cond1, cond2, ..., condk {
             if inSet1(I) then output 1
             elsif inSet2(I) then output 1
                 ...
             elsif inSetk(I) then output 1
             else output 1
         }
    }
\end{verbatim}
\end{center}
}
\item{
A single component of the form

\begin{center}
\begin{verbatim}
    component Base2 provides base1 
                   requires cond1, cond2, ..., condk {
         void main1(Input I) {
             if inSet1(I) then output 1
             elsif inSet2(I) then output 1
                 ...
             elsif inSetk(I) then output 1
             else output 0
         }
    }
\end{verbatim}
\end{center}
}
\item{
A set of $k|V|$ components of the form

\begin{center}
\begin{verbatim}
    component InSetJK provides condJ {
         Boolean inSetJ(Input I) {
             return v_I(x_K) 
         }
    }
\end{verbatim}
\end{center}

\noindent
for $1 \leq {\tt J} \leq k$ and $1 \leq {\tt K} \leq |V|$.
}
\end{enumerate}

\noindent
Note that in $L_{comp}$, there are $|V|$ implementations of each
{\tt cond}-interface.  Finally, let $c$ be component {\tt TopBase} in 
$L_{comp}$, $S$ be the component-based software system composed of
{\tt TopBase}, {\tt Base1}, {\tt Base1A}, and {\tt InSetJJ} for
$1 \leq J \leq k$, and $k' = k + 2$.
Note that the instance of CSAdapt$_{poly}$ described above can be constructed in 
time polynomial in the size of the given instance of {\sc Dominating set};
moreover, as there is only a $(k + 1)$-clause {\tt if-then} statement block and
no loops in the component code and $k \leq |V| < |I|$, any candidate
component-based software system created relative to $L_{int}$, $L_{comp}$, and $c$
runs in time linear in the size of input $I'$.

Let us now verify the correctness of this reduction:

\begin{itemize}
\item{
Suppose that there is a dominating set $D$ of size at most $k$ in the given 
instance of {\sc Dominating set}. We can then construct a component-based 
software system $S'$ consisting of $c$, {\tt Base2}, and the $|D|$ 
{\tt InSet}-components corresponding to the vertices in $D$; the choice of which
interface to implement for each vertex is immaterial, and if there are less than
$k$ vertices in $D$, the final $k - |D|$ required {\tt cond}-interfaces can be 
implemented relative to {\tt InSet}-components corresponding to arbitrary 
vertices in $D$. Observe that for each $(i_j,o_j) \in R$, this software system 
produces output $o_j$ given input $i_j$ and $Rew_{\#comp}(S') = k + 2$.
}
\item{
Conversely, suppose that the constructed instance of CSAdapt$_{poly}$ has a 
\linebreak working 
component-based software system $S'$ based on $c$ relative to $L_{int}$,
$L_{comp}$, and $R$ such that $Rew_{\#comp}(S') \leq k' = k + 2$.  This 
system cannot incorporate components {\tt Base1} and {\tt Base1a}, as this
would result in $Rew_{\#comp}(S') = k + 4 \not\leq k' = k + 2$. Hence,
$S'$ must include component {\tt Base2}. In order to correctly accommodate all
input-output pairs in $R$, the $ \leq k$ {\tt if-then} statements in {\tt Base2}
must implement {\tt InSet}-components whose corresponding vertices form a 
dominating set in $G$ of size at most $k$.
Hence, the existence of a working component-based software system $S'$ for the
constructed instance of CSAdapt$_{poly}$ under $Rew_{\#comp}()$ such that
$Rew_{\#comp}(S') \leq k'$  implies the existence of a
dominating set of size at most $k$ for the given instance  of
{\sc Dominating set}.
}
\end{itemize}

\noindent
This completes the proof of correctness of the reduction.
Given the $W[2]$-hardness of $\la k \ra$-{\sc Dominating set}, this
reduction implies that CSAdapt$_{poly}$ under $Rew_{\#comp()}$ is $W[2]$-hard when
$C_{pi} = 1$, $S_{depth} = 4$, $C_{ri} = k$, and $|L_{int}| = S_{comp} = k + 3$ and 
hence not fp-tractable relative to these parameters unless $FPT = W[2]$.
The fp-intractability result for ESAdapt$_{poly}$ under $Rew_{\#comp}()$ then 
follows by contradiction from Observation \ref{ObsESAdaptSolvability}.
\end{proof}

\begin{description}
\item[{\bf Result B.12}] For any choice of $Env()$, if $\langle |L_{int}|, C_{pi}, C_{ri}, 
                   S_{comp}, S_{depth}\ra$-ESAdapt$_{poly}$ under \\ $Rew_{CodeB}()$ is 
                   fp-tractable then $FPT = W[2]$.
\end{description}
\begin{proof}
Observe that in the reduction in the proof of Result B.11, if
$k' = 6k + 15$, the only possible working component-based software
system is that including {\tt TopBase} and {\tt Base2}. The result
then holds by a modified version of the proof of Result B.11. 
\end{proof}

\begin{description}
\item[{\bf Result B.13}] For any choice of $Env()$, if $\langle I_{ci}, C_{pi}, 
                    S_{depth}\ra$-ESAdapt$_{poly}$ under \linebreak $Rew_{\#comp}()$ is fp-tractable then 
                    $P = NP$.
\end{description}

\begin{description}
\item[{\bf Result B.14}] For any choice of $Env()$, if $\langle I_{ci}, C_{pi}, 
                    S_{depth}\ra$-ESAdapt$_{poly}$ under \linebreak
                    $Rew_{CodeB}()$ is fp-tractable then $P = NP$.
\end{description}

\begin{description}
\item[{\bf Result B.15}] For any choice of $Env()$, if $\langle I_{ci}, C_{pi}, 
                    C_{ri}\ra$-ESAdapt$_{poly}$ under
                   $Rew_{\#comp}()$ is fp-tractable then $P = NP$.
\end{description}

\begin{description}
\item[{\bf Result B.16}] For any choice of $Env()$, if $\langle I_{ci}, C_{pi}, 
                    C_{ri}\ra$-ESAdapt$_{poly}$ under
                   $Rew_{CodeB}()$ is fp-tractable then $P = NP$.
\end{description}

\begin{description}
\item[{\bf Result B.17}] For any choice of $Env()$, if $\langle |L_{comp}|, I_{ci}, 
                    S_{depth}\ra$-ESAdapt$_{poly}$ under \linebreak
                    $Rew_{\#comp}()$ is fp-tractable then $P = NP$.
\end{description}

\begin{description}
\item[{\bf Result B.18}] For any choice of $Env()$, if $\langle |L_{comp}|, I_{ci}, 
                    S_{depth}\ra$-ESAdapt$_{poly}$ under \linebreak
                    $Rew_{CodeB}()$ is fp-tractable then $P = NP$.
\end{description}

\noindent
The next three results follow from the algorithms in Results A.12--A.14,
respectively, in which all candidate working systems are generated and
hence can be evaluated under any $Rew()$ of interest that is computable in
time polynomial in the size of given system $S$.

\begin{description}
\item[{\bf Result B.19}] For any choice of $Rew()$ and $Env()$, $\langle I_{ci}, C_{ri}, 
                    S_{depth}\ra$-ESAdapt$_{poly}$ is fp-tractable.
\end{description}

\begin{description}
\item[{\bf Result B.20}] For any choice of $Rew()$ and $Env()$, $\langle I_{ci}, 
                    S_{comp}\ra$-ESAdapt$_{poly}$ is \linebreak fp-tractable.
\end{description}

\begin{description}
\item[{\bf Result B.21}] For any choice of $Rew()$ and $Env()$, $\langle |L_{int}, 
                    |L_{comp}|\ra$-ESAdapt$_{poly}$ is \linebreak fp-tractable.
\end{description}

\noindent
Given the fp-intractability results we have at this time, none of these 
fp-tractability results are known to be minimal in the sense described at the 
end of Section \ref{SectResESC_FPSolv}. 

\section{Discussion}

\label{SectDisc}

In this section, we will first summarize our results for the problems 
ESCreate and ESAdapt (Section \ref{SectDiscSumm}). We then discuss the 
implications of these results for real-world implementations of emergent 
software systems (Section \ref{SectDiscImpRW}). Several
points raised in this discussion motivate the inclusion of new results;
so as not to disturb the flow of the discussion, proofs when required are
given in the appendix. Finally, we conclude with some general thoughts and 
caveats on how best to interpret and use the results of computational
complexity analyses such as ours within software engineering (Section
\ref{SectDiscCav}).

\subsection{Summary of Results}

\label{SectDiscSumm}

Our results establish that neither of our investigated problems ESCreate or
ESAdapt is solvable exactly by any algorithm when no restrictions are
placed on the structure or operation of the derived software systems 
(Results A.1, B.1, and B.2). This intractability
still holds for both problems relative to polynomial-time exact and exact
promise solvability (Results A.2, A.3, B.3, and B.4) and all three of the
considered forms of polynomial-time approximability (Results A.4--A.7 and 
B.5--B.10) even when software systems are restricted to run in time
polynomial in the sizes of their inputs. Moreover, both problems when so
restricted remain
fixed-parameter intractable relative to all parameters listed in Table
\ref{TabPrm} (Results A.8--A.11 and B.11--B.18), both individually and in 
many combinations and even when many parameters have values that are small
constants. That being said, there are several combinations of parameters
that yield fp-tractability for our problems (see Tables 
\ref{TabPrmResESC}--\ref{TabPrmResESA2}).

\begin{table}[t]
\caption{
A Detailed Summary of Our Parameterized Complexity Results for ESCreate$_{poly}$ and
ESAdapt$_{poly}$ under $Rew_{\#comp}()$ and $Rew_{CodeB}()$. a) Summary for ESCreate$_{poly}$.
Each column in this table is a result which holds relative to the parameter-set 
consisting of all parameters with a @-symbol in that column. If the result 
holds when a particular parameter has a constant value $c$, that is 
indicated by $c$ replacing @  for that parameter in that result's column.
Fp-intractability results are given first and fp-tractability results (shown 
in bold) are given last. 
}
\label{TabPrmResESC}
\vspace*{0.1in}
\centering
\begin{tabular}{ | p{1.1cm}  || p{0.6cm}  p{0.6cm} p{0.6cm}  p{0.6cm}  p{0.6cm}  p{0.6cm}  p{0.6cm}  | }
\hline
& A.8 & A.9 & A.10 
& A.11 & {\bf A.12} & {\bf A.13} & {\bf A.14} \\
\hline\hline
$|L_{int}|$ & @  & -- & -- 
            & -- & {\bf --} & {\bf --} & {\bf @} \\
\hline
$|L_{comp}|$ & -- & -- & -- 
            & 3 & {\bf --} & {\bf --} & {\bf @} \\
\hline\hline
$I_{ci}$    & -- & 2 & 2 
            & 2 & {\bf @} & {\bf @} & {\bf --} \\
\hline\hline
$C_{pi}$    & 1 & 1 & 1 
            & -- & {\bf --} & {\bf --} & {\bf --} \\
\hline
$C_{ri}$    & @ & -- & 2 
            & -- & {\bf @} & {\bf --} & {\bf --} \\
\hline\hline 
$S_{comp}$  & @ & -- & -- 
            & -- & {\bf --} & {\bf @} & {\bf --} \\
\hline
$S_{depth}$ & 2 & 3 & -- 
            & 2 & {\bf @} &  {\bf --} & {\bf --} \\
\hline
\end{tabular}
\end{table}

\begin{table}[t]
\caption{
A Detailed Summary of Our Parameterized Complexity Results for ESCreate$_{poly}$ and
ESAdapt$_{poly}$ under $Rew_{\#comp}()$ and $Rew_{CodeB}()$ (Cont'd). b) Summary for 
ESAdapt$_{poly}$ under $Rew_{\#comp}()$.
}
\vspace*{0.1in}
\label{TabPrmResESA1}
\centering
\begin{tabular}{ | p{1.1cm}  || p{0.6cm}  p{0.6cm} p{0.6cm}  p{0.6cm}  p{0.6cm}  p{0.6cm}  p{0.6cm}  | }
\hline
& B.11 & B.13 & B.15 
& B.17 & {\bf B.19} & {\bf B.20} & {\bf B.21} \\
\hline\hline
$|L_{int}|$ & @  & -- & --  
            & -- & {\bf --} & {\bf --} & {\bf @} \\
\hline
$|L_{comp}|$ & -- & -- & --  
            & 5 & {\bf --} & {\bf --} & {\bf @} \\
\hline\hline
$I_{ci}$    & -- & 2 & 2 
            & 2 & {\bf @} & {\bf @} & {\bf --} \\
\hline\hline
$C_{pi}$    & 1 & 1 & 1  
            & -- & {\bf --} & {\bf --} & {\bf --} \\
\hline
$C_{ri}$    & @ & -- & 2 
            & -- & {\bf @} & {\bf --} & {\bf --} \\
\hline\hline 
$S_{comp}$  & @ & -- & -- 
            & -- & {\bf --} & {\bf @} & {\bf --} \\
\hline
$S_{depth}$ & 3 & 4 & -- 
            & 3 & {\bf @} &  {\bf --} & {\bf --} \\
\hline
\end{tabular}
\end{table}

\begin{table}[t]
\caption{
A Detailed Summary of Our Parameterized Complexity Results for ESCreate$_{poly}$ and
ESAdapt$_{poly}$ under $Rew_{\#comp}()$ and $Rew_{CodeB}()$ (Cont'd). c) Summary for 
ESAdapt$_{poly}$ under $Rew_{CodeB}()$.
}
\vspace*{0.1in}
\label{TabPrmResESA2}
\centering
\begin{tabular}{ | p{1.1cm}  || p{0.6cm}  p{0.6cm} p{0.6cm}  p{0.6cm}  p{0.6cm}  p{0.6cm}  p{0.6cm}  | }
\hline
& B.12 & B.14 & B.16 
& B.18 & {\bf B.19} & {\bf B.20} & {\bf B.21} \\
\hline\hline
$|L_{int}|$ & @  & -- & --  
            & -- & {\bf --} & {\bf --} & {\bf @} \\
\hline
$|L_{comp}|$ & -- & -- & --  
            & 5 & {\bf --} & {\bf --} & {\bf @} \\
\hline\hline
$I_{ci}$    & -- & 2 & 2 
            & 2 & {\bf @} & {\bf @} & {\bf --} \\
\hline\hline
$C_{pi}$    & 1 & 1 & 1  
            & -- & {\bf --} & {\bf --} & {\bf --} \\
\hline
$C_{ri}$    & @ & -- & 2 
            & -- & {\bf @} & {\bf --} & {\bf --} \\
\hline\hline 
$S_{comp}$  & @ & -- & -- 
            & -- & {\bf --} & {\bf @} & {\bf --} \\
\hline
$S_{depth}$ & 3 & 4 & -- 
            & 3 & {\bf @} &  {\bf --} & {\bf --} \\
\hline
\end{tabular}
\end{table}

Our intractability results have a surprisingly broad applicability. This
is because all results (excluding polynomial-time exact promise unsolvability and
cost-inapproximability) are derived relative to two underlying problems --- namely,
CSCreate (which asks for any working software system $S$ relative to the given
interface and component libraries $L_{int}$ and $L_{comp}$ and software system requirements 
$R$) and CSAdapt (which asks for a working software system $S'$ relative to given $L_{int}$,
$L_{comp}$, $R$, and $S$ such that $Rew(S') \leq k$ for some given $k$). The use of
CSCreate and CSAdapt has two consequences. First, as CSCreate does not invoke
$Env()$ or $Rew()$, intractability results derived relative to CSCreate hold for
ESCreate relative to each possible choice of $Env()$ and $Rew()$ and hence ESCreate in
general.  By analogous reasoning, as CSAdapt
does not invoke $Env()$ but does explicitly invoke $Rew()$ relative to given bound $k$,
intractability results derived relative to CSAdapt hold for ESAdapt relative to 
only some possible choices of $Env()$ and $Rew()$ --- namely, any possible $Env()$ and
one of either $Rew_{\#comp}()$ or $Rew_{CodeB}()$. However, as these cases must
be solved by any algorithm that solves ESAdapt with arbitrary input $Env()$ and $Rew()$,
these intractability results also hold for ESAdapt in general. Second, as neither CSCreate 
nor CSAdapt optimizes derived software systems relative to $Rew()$, the observed forms
of intractability for ESCreate and ESAdapt cannot be attributed to optimization of
derived software systems relative to $Rew()$.

\subsection{Implications for Real-world Emergent Software Systems}

\label{SectDiscImpRW}

The foregoing is all well and good for our problems ESAdapt and ESCreate. However, how applicable
are our results to real-world emergent software system creation and adaptation? Given
the use of CSCreate (which ignores $Rew()$) to derive results for ESCreate, all of our 
results, both for tractability and intractability, apply directly to real-world emergent 
system creation. The situation is the same for ESAdapt relative to real-world emergent
system adaptation if one is only given the information in the stated input to ESAdapt.
However, this is not the case in the real-world emergent software system described
in \cite{PG+16,FP17}, where the adaptation process is also given the complete list of working
software systems created relative to $L_{int}$, $L_{comp}$, and $R$ (see Section 
\ref{SectForm}). If this list
is small enough (i.e., of size polynomial in the sizes of $L_{int}$, $L_{comp}$, and $R$),
then a linear scan of this list can determine a candidate with improved (indeed,
optimal) performance relative to $Rew()$ and adaptation can be done in time polynomial 
the sizes of $L_{int}$, $L_{comp}$, $R$, and $S$. 
The issue of whether an intractable problem can have some fixed part of its input (e.g., 
$L_{int}$, $L_{comp}$, $R$, and $c$) preprocessed (probably in non-polynomial time) to create 
information of polynomial size that can be used to solve subsequent instances of the
problem with a varying part (e.g., $S$, $Rew()$, and $k$) in polynomial time is addressed
by the complexity framework given in \cite{CD+02}. At present, relative to an
admittedly artificial and problem-specific reward function, we have
the following result.

\begin{description}
\item[Result B.22]:] If ESAdapt$_{poly}$  can have $L_{int}$, $L_{comp}$, $R$, and $c$ 
                      preprocessed to create polynomial-size information that can be used
                      to solve ESAdapt$_{poly}$  instances of arbitrary $S$, $Rew()$, and 
                      $k$ in polynomial time then the Polynomial Hierarchy $PH$ collapse, 
                      i.e., $PH = \Sigma^p_2$.
\end{description}

\noindent
There is good reason to believe that this result holds, as it is widely believed within
computer science that the Polynomial Hierarchy does not collapse, let alone to a level
as low in the Hierarchy as $\Sigma^p_2$ \cite{BDG88,Gol08,RST15}. Note that this 
result applies to {\em any} polynomial-length information obtained by preprocessing, and 
not just as in \cite{FP17,PG+16} a list of working software systems relative to $L_{int}$,
$L_{comp}$, $c$, and $R$. Though Result B.22 does not invalidate the adaptation 
strategy employed in \cite{FP17,PG+16},
it does suggest that it is not universally applicable to all instances of emergent
system adaptation, and that further complexity-theoretic work is necessary to determine
those situations in which it does and does not work.

Let us now consider some additional implications of our results:

\begin{itemize}
\item{
All of our polynomial-time and fixed-parameter results for runtime-restricted
ESCreate and ESAdapt hold when valid software systems are restricted to run (
and hence can be verified against system requirements) in time polynomial (and indeed,
as noted in our proofs, linear) in the size of the given input. This suggests
that the basic computational difficulty of these problems is tied not so much
to V\&V as the acts of creating working software systems from interface and
component libraries relative to a set of functional system requirements (in the case of
ESCreate) and attempting to merely improve (and not necessarily 
optimize) the performance of an existing working system (in the case of
ESAdapt). 
}
\item{
Our fixed-parameter intractability and tractability results for both of our 
problems relative to the parameter-set $\{I_{ci}, C_{pi}, C_{ri}, S_{depth}, 
S_{comp}\}$ are complete,  in the sense that the fp-status of our problems 
relative to each of the subsets of this parameter set is known (see Tables
\ref{TabSPCA_ESCreate}--\ref{TabSPCA_ESAdapt2}).
Moreover, our fp-tractability results are not only all based on different
runtime analyses of the same brute-force candidate software system
algorithm but are all minimal, in the sense that no subset of the
parameters invoked in any of these fp-tractability results yields
fp-tractability. This suggests that, short of additional restrictions,
the proof-of-concept brute-force enumeration
algorithm proposed in \cite{FP17,PG+16} for creating valid software systems 
may in fact be the best possible.
}
\item{
As all of our intractability results hold when component interfaces support 
specialized rather than universal composability, components are small, 
and components comply with the single-responsibility design pattern, claims that these 
attributes help tame the state space explosion when creating valid software systems from
interface and component libraries (\cite[pages 4 and 5]{FP17} and 
\cite[page 2]{PF19},
respectively) should be seen as incomplete, in the sense that additional
restrictions may be required to account for observed efficient operation.
The same holds relative to claims of the proposed 
linear bandit learning 
algorithm working with high probability \cite[page 340]{PG+16}
given our polynomial-time probabilistic inapproximability results for ESAdapt.
}
\end{itemize}

\noindent
Though the situation may change in
future as additional results are derived, the above does suggest that emergent
software system adaptation may indeed be computationally easier than
emergent software system creation, and that efficient adaptation (after
an initial computationally costly but unavoidable system creation) may be possible in 
emergent software systems. 

The above (in  particular, the second and third bullet-points)
very much begs the question of under what additional restrictions such 
efficient adaptation or creation might be possible. Answering this question will aid
both designing the best possible emergent software systems and explaining when and why 
existing systems do perform well (e.g., the claims listed in point (3) above). Known sets of
aspect restrictions
that yield fp-tractability like those reported in this paper are a start, but given the 
observed prevalence of fp-intractability, more aspects need to be analyzed. A good source
for these would be restrictions that ``break'' our intractability reductions, e.g.,
restricting the number of instances of required interfaces in a component that can be 
implemented by multiple components. One should also not underestimate the value of 
additional broad restrictions
like our requirement that derived software systems run in polynomial time. That
such restrictions are critical is demonstrated by the three results below, which follow
from the reductions in the proofs of Results A.1, B.1, and B.2 (in which the
values of $|L_{int}|, |L_{comp}|, I_{ci}, C_{pi}, C_{ri}, S_{comp}$, and $S_{depth}$ are
all very small constants) and the definition of fp-tractability.

\begin{description}
\item[{\bf Result A.15}] For any choice of $Rew()$ and $Env()$, 
$\la |L_{int}|, |L_{comp}|, I_{ci}, C_{pi}, C_{ri}, S_{comp}, \linebreak S_{depth}\ra$-ESCreate
is unsolvable.
\item[{\bf Result B.23}] For any choice of $Env()$, 
$\la |L_{int}|, |L_{comp}|, I_{ci}, C_{pi}, C_{ri}, S_{comp}, \linebreak S_{depth}\ra$-ESAdapt
under $Rew_{\#comp}()$ is unsolvable.
\item[{\bf Result B.24}] For any choice of $Env()$, 
$\la |L_{int}|, |L_{comp}|, I_{ci}, C_{pi}, C_{ri}, S_{comp}, \linebreak S_{depth}\ra$-ESAdapt
under $Rew_{CodeB}()$ is unsolvable.
\end{description}

\noindent
Another demonstration is the above-noted solvability of ESAdapt in polynomial time
if polynomial-size preprocessed information is available. Additional possible restrictions of
this type that might be useful could be on the forms of environment or reward functions or
the internal code structure and/or maximum code-length of components. It might also be of 
use to consider alternative problem formulations. For example, requiring that ESAdapt only
derive a new system whose $Rew()$-performance is a constant-valued improvement over
that of the given system $S$ (cf., the absolute rather than $S$-relative performance
bound currently implemented by $k$ in CSAdapt) will not lower the computational complexity
of ESAdapt relative to $Rew_{\#comp}()$ (as the twinning trick used to derive these 
intractability results means that working systems differ in at most one component and hence 
in $Rew_{\#comp}()$-value by at most one) but it may well work relative to other reward 
functions.

\begin{table}[t]
\caption{Systematic Parameterized Complexity Analysis of ESCreate$_{poly}$
          and ESAdapt$_{poly}$ Relative to the Parameter-set $\{I_{ci}, C_{pi}, C_{ri},
          S_{depth}, S_{comp}\}$. a) Analysis for ESCreate$_{poly}$. In this table, a 
          $\surd$ (X) symbol indicates that ES-create relative to the 
          parameter-set composed of the union of the row and column
          parameter-sets indexing that entry is fp-(in)tractable.
          The original results are subscripted with the result-number
          followed by a star; all other results are subscripted by
          number of the result they are derived from relative to
          Lemmas \ref{LemPrmProp1} and \ref{LemPrmProp2}.
}
\label{TabSPCA_ESCreate}
\vspace*{0.1in}
\centering
\begin{tabular}{| c || c | c | c | c |}
\hline
    & --- & $S_{comp}$ & $S_{depth}$ & $S_{comp}, S_{depth}$ \\
\hline\hline
---                      & NPh & X$_{A.8}$ & X$_{A.8}$ & X$_{A.8}$ \\
\hline
$I_{ci}$                 & X$_{A.9}$ & $\surd_{A.13*}$ & X$_{A.9}$ & $\surd_{A.13}$ \\
\hline
$C_{pi}$                 & X$_{A.8}$ & X$_{A.8}$ & X$_{A.8}$ & X$_{A.8}$ \\ 
\hline
$C_{ri}$                 & X$_{A.8}$ & X$_{A.8}$ & X$_{A.8}$ & X$_{A.8}$ \\ 
\hline
$I_{ci}, C_{pi}$         & X$_{A.9}$ & $\surd_{A.13}$ & X$_{A.9*}$ & $\surd_{A.13}$ \\
\hline
$I_{ci}, C_{ri}$         & X$_{A.10}$ & $\surd_{A.13}$ & $\surd_{A.12*}$ & $\surd_{A.12}$ \\
\hline
$C_{pi}, C_{ri}$         & X$_{A.8}$ & X$_{A.8}$ & X$_{A.8}$ & X$_{A.8*}$ \\ 
\hline
$I_{ci}, C_{pi}, C_{ri}$ & X$_{A.10*}$ & $\surd_{A.13}$ & $\surd_{A.12}$ & $\surd_{A.12}$ \\
\hline
\end{tabular}
\end{table}

\begin{table}[t]
\caption{Systematic Parameterized Complexity Analysis of ESCreate$_{poly}$
          and ESAdapt$_{poly}$ Relative to the Parameter-set $\{I_{ci}, C_{pi}, C_{ri},
          S_{depth}, S_{comp}\}$ (Cont'd). b) Analysis for ESAdapt$_{poly}$ under
          $Rew_{\#comp}()$.
}
\label{TabSPCA_ESAdapt1}
\vspace*{0.1in}
\centering
\begin{tabular}{| c || c | c | c | c |}
\hline
    & --- & $S_{comp}$ & $S_{depth}$ & $S_{comp}, S_{depth}$ \\
\hline\hline
---                      & NPh & X$_{B.11}$ & X$_{B.11}$ & X$_{B.11}$ \\
\hline
$I_{ci}$                 & X$_{B.13}$ & $\surd_{B.20*}$ & X$_{B.13}$ & $\surd_{B.20}$ \\
\hline
$C_{pi}$                 & X$_{B.11}$ & X$_{B.11}$ & X$_{B.11}$ & X$_{B.11}$ \\ 
\hline
$C_{ri}$                 & X$_{B.11}$ & X$_{B.11}$ & X$_{B.11}$ & X$_{B.11}$ \\ 
\hline
$I_{ci}, C_{pi}$         & X$_{B.13}$ & $\surd_{B.20}$ & X$_{B.13*}$ & $\surd_{B.20}$ \\
\hline
$I_{ci}, C_{ri}$         & X$_{B.15}$ & $\surd_{B.20}$ & $\surd_{B.19*}$ & $\surd_{B.19}$ \\
\hline
$C_{pi}, C_{ri}$         & X$_{B.11}$ & X$_{B.11}$ & X$_{B.11}$ & X$_{B.11*}$ \\ 
\hline
$I_{ci}, C_{pi}, C_{ri}$ & X$_{B.15*}$ & $\surd_{B.20}$ & $\surd_{B.19}$ & $\surd_{B.19}$ \\
\hline
\end{tabular}
\end{table}

\begin{table}[t]
\caption{Systematic Parameterized Complexity Analysis of ESCreate$_{poly}$
          and ESAdapt$_{poly}$ Relative to the Parameter-set $\{I_{ci}, C_{pi}, C_{ri},
          S_{depth}, S_{comp}\}$ (Cont'd). c) Analysis for ESAdapt$_{poly}$ under
          $Rew_{CodeB}()$.
}
\label{TabSPCA_ESAdapt2}
\vspace*{0.1in}
\centering
\begin{tabular}{| c || c | c | c | c |}
\hline
    & --- & $S_{comp}$ & $S_{depth}$ & $S_{comp}, S_{depth}$ \\
\hline\hline
---                      & NPh & X$_{B.12}$ & X$_{B.12}$ & X$_{B.12}$ \\
\hline
$I_{ci}$                 & X$_{B.14}$ & $\surd_{B.20*}$ & X$_{B.14}$ & $\surd_{B.20}$ \\
\hline
$C_{pi}$                 & X$_{B.12}$ & X$_{B.12}$ & X$_{B.12}$ & X$_{B.12}$ \\ 
\hline
$C_{ri}$                 & X$_{B.12}$ & X$_{B.12}$ & X$_{B.12}$ & X$_{B.12}$ \\ 
\hline
$I_{ci}, C_{pi}$         & X$_{B.14}$ & $\surd_{B.20}$ & X$_{B.14*}$ & $\surd_{B.20}$ \\
\hline
$I_{ci}, C_{ri}$         & X$_{B.16}$ & $\surd_{B.20}$ & $\surd_{B.19*}$ & $\surd_{B.19}$ \\
\hline
$C_{pi}, C_{ri}$         & X$_{B.12}$ & X$_{B.12}$ & X$_{B.12}$ & X$_{B.12*}$ \\ 
\hline
$I_{ci}, C_{pi}, C_{ri}$ & X$_{B.16*}$ & $\surd_{B.20}$ & $\surd_{B.19}$ & $\surd_{B.19}$ \\
\hline
\end{tabular}
\end{table}

\subsection{Computational Complexity Analysis and \\ Software Engineering}

\label{SectDiscCav}

Quite aside from the utility of intractability results for ESCreate and
ESAdapt in
answering questions about the operation of existing emergent software 
systems and the design of new ones, the mere existence of such results 
relative to standard types of intractability such as Turing unsolvability
and $NP$- and $W$-hardness has additional far-reaching consequences. This is
because of structural complexity theory, the subdiscipline of theoretical computer
science that investigates the (non)inclusion relationships between various complexity
classes \cite{BDG88,BDG90,CompZoo,Joh90}. As complexity classes are typically
defined relative to particular types of algorithms (e.g., $P$, $EXP$, and $PSPACE$
are the classes of problems solvable in polynomial time, exponential time, and
polynomial space, respectively), standard intractability results typically imply many other 
intractability results. For example, given the imminent availability
of functional quantum computers and the accompanying hoopla \cite{GI19}, one might think 
that polynomial-time intractability (in particular, $NP$-hardness) results like those
given here are moot. However, it is widely believed in computer science that $NP 
\not\subseteq BQP$, where $BQP$ (Bounded-error Quantum Polynomial Time) is considered the 
largest class of problems solvable by usable quantum computer algorithms \cite{BV97}. If any 
$NP$-hard problem is in $BQP$ then this conjecture is false \cite{BV97}, rendering such 
quantum solvability of $NP$-hard problems like ours very unlikely. 

The above, in combination
with the points raised and implications described earlier in Section
\ref{SectDisc}, highlights why
computational complexity analyses are useful in software engineering and moreover the 
context in which they can be used most productively --- namely, as part of an ongoing 
dialogue between software engineers and theoretical computer scientists, in which questions 
raised by members of one group, e.g., 

\begin{itemize}
\item Why do our systems (not) work as well as they do? 
\item Does this restriction on our systems matter, and if so, how?
\item How do we formulate relevant problems for analysis? 
\item What types of algorithms do (not) exist for our problems?
\end{itemize}

\noindent
inspire both investigations by and new questions from the other.

Before we close out this subsection, two final caveats are in order. First,
it is important to note that the brute-force search algorithms underlying all 
of our fp-tractability results are not immediately usable in real-world 
emergent software systems because the running times of these algorithms are 
(to be blunt, ludicrously) exorbitant. This difficulty is also not as bad as it
initially seems because it is well known within the parameterized complexity 
research community that once fp-tractability is proven relative to a 
parameter-set, surprisingly effective fp-algorithms can often be subsequently 
developed \cite{CF+15, FL+19}. This may involve either applying
advanced fp-algorithm design techniques or adding additional parameters to 
minimal parameter-sets, and typically results in algorithm runtimes
with greatly diminished non-polynomial terms and polynomial terms that are 
additive rather than multiplicative.

Second, and perhaps more importantly, the fixed-parameter (and indeed all
of the theoretical) analyses in this paper are only 
intended to sketch out what types of efficient algorithms
do and do not exist for our problems of interest and are not intended to be 
of immediate use. Given this, one may be tempted to conclude that our results 
and by proxy theoretical analyses in general are irrelevant. We believe
that such a view is short-sighted at best and dangerous at worst. Not knowing 
the precise conditions under which existing emergent software system creation 
and adaptation algorithms work 
well may have serious consequences, e.g., drastically slowed software creation 
time and/or unreliable software operation, if these conditions are violated. 
These consequences would be particularly damaging in the case of the fully 
automatic operations underlying emergent software systems. Given that reliable 
software operation is crucial and efficient software creation and adaptation is
at the very least desirable, the acquisition of such knowledge via a combination
of rigorous empirical and theoretical analyses should be a priority. With 
respect to theoretical analyses, it is our hope that the techniques and results
in this paper comprise a useful first step.

\section{Conclusions and Future Work}

\label{SectConc}

In this paper, we have applied computational complexity analysis to
evaluate algorithmic options for emergent software system creation and
adaptation relative to several popular types of exact and approximate
efficient solvability. We have shown that neither problem is correctly
and exactly solvable for all inputs when no restrictions are placed
on the structure and operation of valid software systems.
This intractability continues to hold 
relative to all examined types of efficient exact and approximate solvability
when valid software systems are restricted to run in times polynomial in
their input sizes. Moreover, both of our problems remain intractable under
a variety of additional restrictions on valid software system structure,
both individually and in many combinations. That being said, we give
sets of additional restrictions that do yield tractability for both
problems, as well as circumstantial evidence that emergent software
system adaptation is computationally easier than emergent software
system creation.

There are three promising directions for future research. The first of
these is to extend our current analyses by (1) establishing the fp-status
of our problems relative to all parameter-combinations from Table
\ref{TabPrm} and (2) considering additional parameters for and restrictions 
on emergent software systems. A productive source for the latter might
be restrictions which ``break'' the reduction-tricks used in our
intractability proofs. 
The second research direction is to consider additional computational problems
associated with emergent software systems. Of particular interest here is the
automated synthesis of both variants of existing components and new 
components from input-output examples \cite{FWP19,MWP18}. Analyses of these
problems may also give insights into the more general problem of program
synthesis \cite{BG+17,Gul10,GPS17}.
The third research direction involves 
not so much emergent software systems but self-adaptive and component-based
software systems in general. Invoking the terminology in the taxonomy
of component models given in \cite{CS+11}, emergent software systems
have operation-based interfaces, provide horizontal and partial vertical 
component binding during composition, and explicitly distinguish between the 
required and provided interfaces of a component. These characteristics
are not only exploited by our intractability result reductions but are also
typical of many other component models \cite[Table 2]{CS+11}. Hence, it
would be of great interest to see if our analyses for emergent software systems
could be extended to software system creation and adaptation problems
arising relative to other proposed types of self-adaptive software systems 
\cite{CdL+09,dLG+17,dLG+13} and within component-based software engineering in 
general \cite{GMS02,HC00,VC+16}.

\section*{Acknowledgements}
The authors thank Barry Porter for invaluable discussions about
and comments on previous versions of this paper; though we have not been
able to address all of his concerns, we very much appreciate his efforts
in helping us create a better and more relevant paper.
TW would also like to thank Antonina Kolokolova for useful discussions on 
Rice's Theorem and Ulrike Stege for
introducing him to the research area of self-adaptive software systems.
TW was supported by National Science and Engineering
Council (NSERC) Discovery Grant 228104-2015.


\appendix

\section{Proofs of Selected Results}

Let us first prove Result A.9.
In the reduction below, instead of encoding a dominating set of size $k$ implicitly in the 
components implementing procedures {\tt InSet1(), InSet2(), \ldots, InSetk()} as in the
reductions in the proofs of Lemmas \ref{LemRed_DS_CSCreate1} and 
\ref{LemRed_DS_CSCreateOPT1}, a candidate dominating set
is encoded explicitly in binary-valued vector {\tt vS} in component {\tt Base} in the
components used to implement the procedures {\tt vertexStatus1(), vertexStatus2(),
\ldots,} {\tt  vertexStatus(|V|()}. Moreover, checks to see if this candidate is 
an actual dominating set of size $k$ in $G$ are explicitly coded in component {\tt Base}
(see Figure \ref{FigRed4}. This allows us to reduce the maximum number of interfaces 
provided by any component and the maximum number of components implemented by any interface 
to 1 and 2, respectively, while maintaining a maximum component depth of 3.

\begin{figure}[t]
\begin{center}
\includegraphics[width=4.2in]{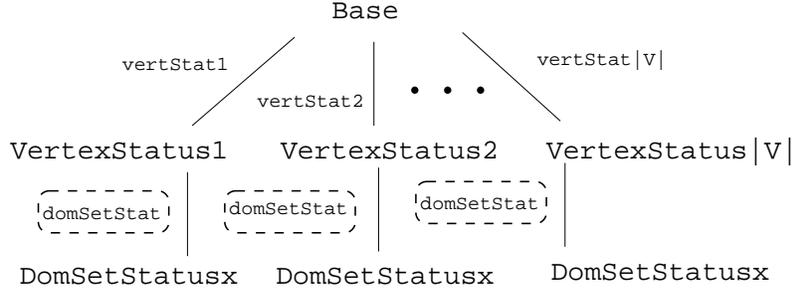}
\end{center}
\caption{General structure of valid software systems created by the reduction 
          in the proof of Result A.9. Note that index $x$ 
          in {\tt DomSetStatusx} is such that $x \in \{0,1\}$. Following the convention
          in Figure \ref{FigExESS3}, interfaces with multiple implementing
          components are enclosed in dashed boxes.
}
\label{FigRed4}
\end{figure}

\begin{description}
\item[{\bf Result A.9}] For any choice of $Rew()$ and $Env()$, if $\langle I_{ci}, C_{pi}, 
                   S_{depth}\ra$-ESCreate$_{poly}$ is fp-tractable then 
                   $P = NP$.
\end{description}
\begin{proof}
Consider the following polynomial-time Karp reduction from {\sc Dominating
set} to CSCreate$_{poly}$:
given an instance $I = \la G = (V,E), k\ra$ of {\sc Dominating set}, construct 
an instance $I' = \la R, L_{int}, L_{comp}, c\ra$ of CSCreate$_{poly}$ in which
$R$ is the same as that given in the reduction in the proof of Lemma
\ref{LemRed_DS_CSCreate1}. Let $L_{int}$ consist of $|V| + 2$
interfaces broken into three groups:

\begin{enumerate}
\item{
A single interface of the form

\begin{center}
\begin{verbatim}
    interface base {
         void main(Input I)
    }
\end{verbatim}
\end{center}
}
\item{
A set of $|V|$ interfaces of the form

\begin{center}
\begin{verbatim}
    interface vertStatJ {
         int vertexStatusJ()
    }
\end{verbatim}
\end{center}

\noindent
for $1 \leq {\tt J} \leq |V|$.
}
\item{
An interface of the form

\begin{center}
\begin{verbatim}
    interface domSetStat {
         int domSetStatus()
    }
\end{verbatim}
\end{center}
}
\end{enumerate}

\noindent
Let $L_{comp}$ consist of $|V| + 3$ components broken into three groups:

\begin{enumerate}
\item{
A single component of the form

\begin{center}
\begin{verbatim}
    component Base provides base 
                   requires vertStat1, vertStat2, ..., 
                            vertStat|V| {
         void main(Input I) {

             create integer array vS of length |V|

             vS[1] = vertexStatus1()
             vS[2] = vertexStatus2()
                 ...
             vS[|V|] = vertexStatus|V|()

             numFound = 0
             for i = 1 to |V| do
                 if vS[j] == 1 then
                       numFound = numFound + 1
             if numFound == k then
                 isCandidatekDomSet = True
             else
                 isCandidatekDomSet = False

             if isCandidatekDomSet then
                 numFound = 0
                 for i = 1 to |V| do
                     if v_I(x_i) and vS[i] == 1 then
                         numFound = numFound + 1

                 if numFound > 0 then
                     output 1
                 else
                     output 0
              else
                 output 0
         }
    }
\end{verbatim}
\end{center}
}
\item{
A set of $|V|$ components of the form

\begin{center}
\begin{verbatim}
    component VertexStatusJ provides vertStatJ {
                          requires domSetStat
         int vertexStatusJ() {
             return domSetStatus()
         }
    }
\end{verbatim}
\end{center}

\noindent
for $1 \leq {\tt J} \leq |V|$.
}
\item{
Two components of the form

\begin{center}
\begin{verbatim}
    component DomSetStatus0 provides domSetStat {
         int domSetStatus() {
             return 0
         }
    }

    component DomSetStatus1 provides domSetStat {
         int domSetStatus() {
             return 1
         }
    }
\end{verbatim}
\end{center}
}
\end{enumerate}

\noindent
Observe that 
in component {\tt Base}, if entry $i$ of array {\tt vS} has value
$1$, then vertex $i$ is in the candidate 
dominating set for $G$ encoded in {\tt vS}. Finally, let $c$ be component {\tt Base}
in $L_{comp}$. Note that the instance of CSCreate$_{poly}$ described
above can be constructed in time polynomial in the size of the given
instance of {\sc Dominating set};
moreover, as there is only a $|V|$-length assignment statement block and
two single-level loops in the component code that
each execute at most $|V| < |I|$ times, any component-based software system 
created relative to $L_{int}$, $L_{comp}$, and $c$
runs in time linear in the size of input $I'$.

Let us now verify the correctness of this reduction:

\begin{itemize}
\item Suppose there is a dominating set $D$ of size at most $k$ in the given
       instance of {\sc Dominating set}; if $|D| < k$, augment $D$ with
       $k - |D|$ arbitrary other vertices from $G$ such that $|D| = k$. 
       Construct a component-based
       software system consisting of $c$ and all $|V|$ components
       {\tt VertexStatusJ} for $1 \leq {\tt J} \leq |V|$ in which the
       {\tt domSetStat} interface in component {\tt VertexStatusJ} is
       implemented by component {\tt DomSetStatus1} if vertex {\tt J}
       is in $D$ and by component {\tt DomSetStatus0}
       otherwise. Observe that for each $(i_j, o_j) \in R$, this
       software system produces output $o_j$ given input $i_j$.
\item Conversely, suppose that the constructed instance of CSCreate$_{poly}$
       has a working component-based software system based on $c$
       relative to $L_{int}$, $L_{comp}$, and $R$. In order to
       accommodate all input-output pairs in $R$, the $|V|$ 
       {\tt vertexStatus} components must implement their
       {\tt domSetStat} interfaces such that the $1$-values in array 
       {\tt vS} specify a dominating set of size $k$ in $G$. Hence,
       the existence of a working component-based software system
       for the constructed instance of CSCreate$_{poly}$ implies the existence
       of a dominating set of size $k$ for the given instance of
       {\sc Dominating set}.
\end{itemize}

\noindent
The reduction above is thus correct. Given the $NP$-hardness of {\sc Dominating set}, 
this reduction implies that CSCreate$_{poly}$ is $NP$-hard when $I_{ci} = 2$, $C_{pi} = 1$, and 
$S_{depth} = 3$ and hence by Lemma \ref{LemPrmProp3} not fp-tractable relative to these 
parameters unless $P = NP$. The fp-intractability result for ESCreate$_{poly}$ then 
follows by contradiction from Observation \ref{ObsESCreateSolvability}.
\end{proof}

\begin{figure}[t]
\begin{center}
\includegraphics[width=3.5in]{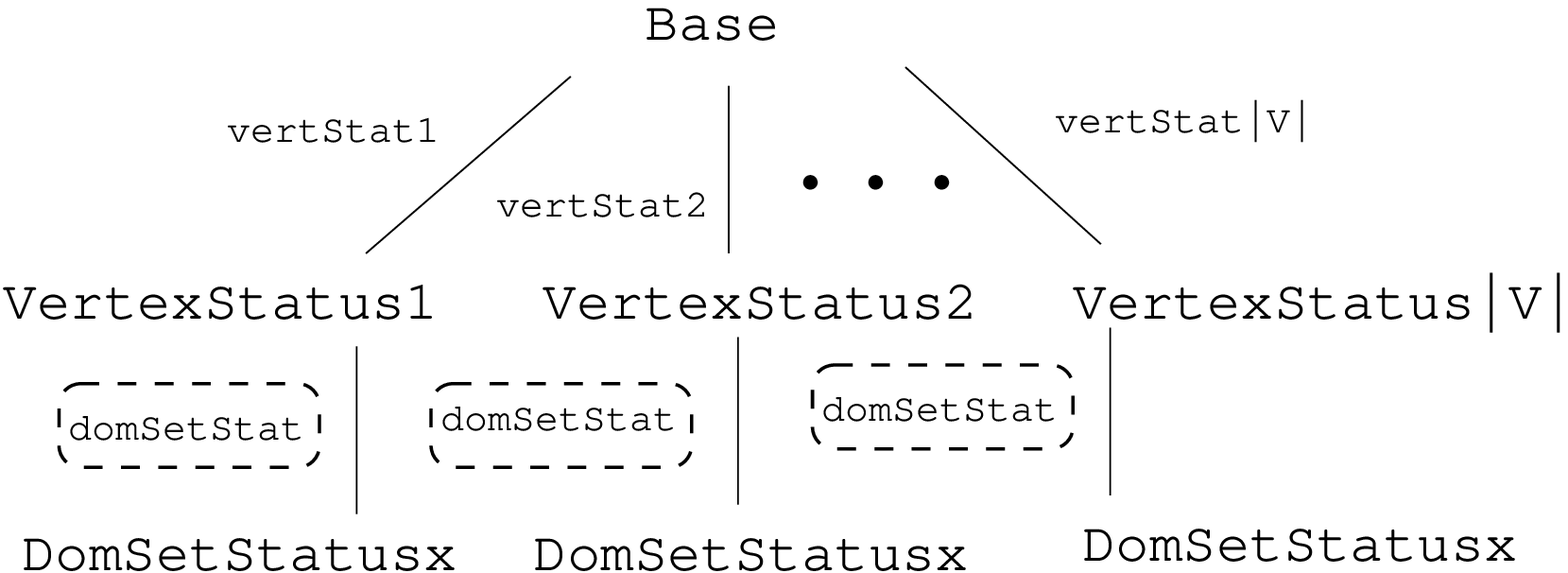}
\end{center}
\caption{General structure of valid software systems created by the reduction 
          in the proof of Result A.10. Note that index $x$ 
          in {\tt DomSetStatusx} is such that $x \in \{0,1\}$. Following the convention
          in Figure \ref{FigExESS3}, interfaces with multiple implementing
          components are enclosed in dashed boxes.
}
\label{FigRed5}
\end{figure}

\vspace*{0.15in}

Let us now prove Result A.10.
In the reduction below, we modify the reduction in the proof of Result A.9 to split the
call to procedures {\tt vertexStatus1(), vertexStatus2(), \ldots,} 
{\tt vertexStatus|V|()} in
component {\tt Base} over $|V|$ base-components (see Figured \ref{FigRed5}). This still 
allows us to have the values of the maximum number of interfaces provided by any component 
and the maximum number of components implemented by any interface as 1 and 2, respectively, 
as in the previous reduction but trades off constant-valued component depth for a 
constant-valued maximum number of interfaces required by any component.

\begin{description}
\item[{\bf Result A.10}] For any choice of $Rew()$ and $Env()$, if $\langle I_{ci}, C_{pi}, 
                    C_{ri}\ra$-ESCreate$_{poly}$ is fp-tractable then $P = NP$.
\end{description}
\begin{proof}
Consider the following polynomial-time Karp reduction from {\sc Dominating
set} to CSCreate$_{poly}$:
given an instance $I = \la G = (V,E), k\ra$ of {\sc Dominating set}, construct 
an instance $I' = \la R, L_{int}, L_{comp}, c\ra$ of CSCreate$_{poly}$ in which
$R$ is the same as that given in the reduction in the proof of Lemma
\ref{LemRed_DS_CSCreate1}. Let $L_{int}$ consist of $2|V| + 1$
interfaces broken into four groups:

\begin{enumerate}
\item{
A single interface of the form

\begin{center}
\begin{verbatim}
    interface base {
         void main(Input I)
    }
\end{verbatim}
\end{center}
}
\item{
A set of $|V| - 1$ interfaces of the form

\begin{center}
\begin{verbatim}
    interface baseJ {
         void callBaseJ(Input I, int[] vS)
    }
\end{verbatim}
\end{center}

\noindent
for $2 \leq {\tt J} \leq |V|$.
}
\item{
A set of $|V|$ interfaces of the form

\begin{center}
\begin{verbatim}
    interface vertStatJ {
         int vertexStatusJ()
    }
\end{verbatim}
\end{center}

\noindent
for $1 \leq {\tt J} \leq |V|$.
}
\item{
An interface of the form

\begin{center}
\begin{verbatim}
    interface domSetStat {
         int domSetStatus()
    }
\end{verbatim}
\end{center}
}
\end{enumerate}

\noindent
Let $L_{comp}$ consist of $2|V| + 2$ components broken into five groups:

\begin{enumerate}
\item{
A single component of the form

\begin{center}
\begin{verbatim}
    component Base provides base 
                   requires vertStat1, base2 {
         void main(Input I) {

             create integer array vS of length |V|

             vS[1] = vertexStatus1()

             callBase2(I, vS)
         }
    }
\end{verbatim}
\end{center}
}
\item{
A set of $|V| - 2$ components of the form

\begin{center}
\begin{verbatim}
    component BaseJ provides baseJ 
                   requires vertStatJ, base(J + 1) {
         void callBaseJ(Input I, int[] vS) {

             vS[J] = vertexStatusJ()

             callBase(J + 1)(I, vS)
         }
    }
\end{verbatim}
\end{center}

\noindent
for $2 \leq {\tt J} \leq |V| - 1$.
}
\item{
A single component of the form

\begin{center}
\begin{verbatim}
    component Base|V| provides base|V| 
                   requires vertStat|V| {
         void callBase|V|(Input I, int[] vS) {

             vS[|V|] = vertexStatus|V|()

             numFound = 0
             for j = 1 to |V| do
                 if vS[i] == i then
                    numFound = numFound + 1
             if numFound = k then
                 isCandidatekDomSet = True
             else
                 isCandidatekDomSet = False

             if isCandidatekDomSet then
                 numFound = 0
                 for i = 1 to |V| do
                     if v_I(x_i) and vS[i] == 1 then
                         numFound = numFound + 1

                 if numFound > 0 then
                     output 1
                 else
                     output 0
             else
                 output = 0
         }
    }
\end{verbatim}
\end{center}
}
\item{
A set of $|V|$ components of the form

\begin{center}
\begin{verbatim}
    component VertexStatusJ provides vertStatJ {
                          requires domSetStat
         int vertexStatusJ() {
             return domSetStatus()
         }
    }
\end{verbatim}
\end{center}

\noindent
for $1 \leq {\tt J} \leq |V|$.
}
\item{
Two components of the form

\begin{center}
\begin{verbatim}
    component DomSetStatus0 provides domSetStat {
         int domSetStatus() {
             return 0
         }
    }

    component DomSetStatus1 provides domSetStat {
         int domSetStatus() {
             return 1
         }
    }
\end{verbatim}
\end{center}
}
\end{enumerate}

\noindent
Finally, let $c$ be component {\tt Base} in $L_{comp}$. Note that the instance 
of CSCreate$_{poly}$ described above can be constructed in time polynomial in the size 
of the given instance of {\sc Dominating set};
moreover, as there are only single-level loops in the component code that
each execute at most $|V| < |I|$ times, any component-based software system 
created relative to $L_{int}$, $L_{comp}$, and $c$
runs in time linear in the size of input $I'$.

Observe that interfaces
{\tt base}, {\tt vertStatJ} for $1 \leq {\tt J} \leq |V|$, and {\tt domSetStat}
and components {\tt VertexStatusJ} for $1 \leq {\tt J} \leq |V|$, 
{\tt DomSetStatus0}, and {\tt DomSetStatus1} are the same as in the
reduction in the proof of Result A.9. Moreover, interfaces 
{\tt base} and {\tt baseJ}, and components {\tt Base} and {\tt callBaseJ}, 
$2 \leq {\tt J} \leq |V|$, effectively simulate interface {\tt base} and
component {\tt Base} in the same reduction. Hence, the proof of
correctness of the reduction in the proof of Result A.9 can, with
slight modifications, also proves the correctness of the reduction described above.
Given the $NP$-hardness of {\sc Dominating set}, this reduction 
implies that CSCreate$_{poly}$ is $NP$-hard when $I_{ci} = 2$, $C_{pi} = 1$, and $C_{ri} = 2$, 
and hence by Lemma \ref{LemPrmProp3} not fp-tractable relative to these parameters 
unless $P = NP$. The fp-intractability result for ESCreate$_{poly}$ then follows by 
contradiction from Observation \ref{ObsESCreateSolvability}.
\end{proof}

\begin{figure}[t]
\begin{center}
\includegraphics[width=4.2in]{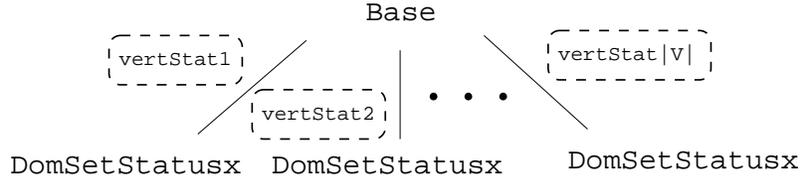}
\end{center}
\caption{General structure of valid software systems created by the reduction 
          in the proof of Result A.11. Note that index $x$ 
          in {\tt DomSetStatusx} is such that $x \in \{0,1\}$. Following the convention
          in Figure \ref{FigExESS3}, interfaces with multiple implementing
          components are enclosed in dashed boxes.
}
\label{FigRed7}
\end{figure}

\vspace*{0.15in}

Let us now prove Result A.11.
The reduction below modifies that given in the proof of Result A.9 such that there 
are still only two components {\tt DomSetStatus0} and {\tt DomSetStatus1} but now, each of 
these components contains customized versions of procedures {\tt vertexStatus1(), 
vertexStatus2(), \ldots, vertexStatus|V|()} instead of placing these procedures separately in
components
{\tt VertexStatus1, \linebreak VertexStatus2, \ldots, VertexStatus|V|} (see Figure \ref{FigRed7}).
By invoking the abilities of different interfaces to implement different copies of the
same component and implementing components to provide to an interface only that 
code which is used by that interface, not only is the maximum component depth reduced
by one but the size of the component library $L_{comp}$ is reduced to a constant
(namely, 3).

\begin{description}
\item[{\bf Result A.11}] For any choice of $Rew()$ and $Env()$, if  $\langle |L_{comp}|, 
                    I_{ci}, S_{depth}\ra$-ESCreate$_{poly}$ is fp-tractable then 
                    $P = NP$.
\end{description}
\begin{proof}
Consider the following polynomial-time Karp reduction from {\sc Dominating
set} to CSCreate$_{poly}$:
given an instance $I = \la G = (V,E), k\ra$ of {\sc Dominating set}, construct 
an instance $I' = \la R, L_{int}, L_{comp}, c\ra$ of CSCreate$_{poly}$ in which
$R$ is the same as that given in the reduction in the proof of Lemma
\ref{LemRed_DS_CSCreate1}. Let $L_{int}$ consist of $|V| + 1$
interfaces broken into two groups:

\begin{enumerate}
\item{
A single interface of the form

\begin{center}
\begin{verbatim}
    interface base {
         void main(Input I)
    }
\end{verbatim}
\end{center}
}
\item{
A set of $|V|$ interfaces of the form

\begin{center}
\begin{verbatim}
    interface vertStatJ {
         int vertexStatusJ()
    }
\end{verbatim}
\end{center}

\noindent
for $1 \leq {\tt J} \leq |V|$.
}
\end{enumerate}

\noindent
Let $L_{comp}$ consist of $3$ components broken into three groups:

\begin{enumerate}
\item{
A single component of the form

\begin{center}
\begin{verbatim}
    component Base provides base 
                   requires vertStat1, vertStat2, ..., 
                            vertStat|V| {
         void main(Input I) {

             create integer array vS of length |V|

             vS[1] = vertexStatus1()
             vS[2] = vertexStatus2()
                 ...
             vS[|V|] = vertexStatus|V|()

             numFound = 0
             for i = 1 to |V| do
                 if vS[i] == i then
                    numFound = numFound + 1
             if num_found == k then
                isCandidatekDomSet = True
             else
                isCandidatekDomSet = False

             if isCandidatekDomSet then
                 numFound = 0
                 for i = 1 to |V| do
                     if v_I(x_i) and vS[i] == 1 then
                         numFound = numFound + 1

                 if numFound > 0 then
                     output 1
                 else
                     output 0
              else
                 output 0
         }
    }
\end{verbatim}
\end{center}
}
\newpage
\item{
Two components of the form

\begin{center}
\begin{verbatim}
    component DomSetStatus0 provides vertStat1, vertStat2, ...
                                     vertStat|V| {

         int vertexStatus1() {
             return 0
         }

         int vertexStatus2() {
             return 0
         }

         ...

         int vertexStatus|V|() {
             return 0
         }
    }

    component DomSetStatus1 provides vertStat1, vertStat2, ...
                                     vertStat|V| {

         int vertexStatus1() {
             return 1
         }

         int vertexStatus2() {
             return 1
         }

         ...

         int vertexStatus|V|() {
             return 1
         }
    }
\end{verbatim}
\end{center}
}
\end{enumerate}

\noindent
Observe that 
in component {\tt Base}, if entry $i$ of array {\tt vS} has value
$1$ then vertex $i$ is in the candidate 
dominating set for $G$ encoded in {\tt vS}. Finally, let $c$ be component {\tt Base}
in $L_{comp}$. Note that the instance of CSCreate$_{poly}$ described
above can be constructed in time polynomial in the size of the given
instance of {\sc Dominating set};
moreover, as there are only a single $|V|$-length assignment block and two
single-level loops in the component code that
each execute at most $|V| < |I|$ times, any component-based software system 
created relative to $L_{int}$, $L_{comp}$, and $c$
runs in time linear in the size of input $I'$.

Let us now verify the correctness of this reduction:

\begin{itemize}
\item Suppose there is a dominating set $D$ of size at most $k$ in the given
       instance of {\sc Dominating set}; if $|D| < k$, augment $D$ with
       $k - |D|$ arbitrary other vertices from $G$ such that $|D| = k$. 
       Construct a component-based
       software system consisting of $c$ and $|V|$ copies of the
       {\tt domSetStatus} components
       in which the
       {\tt vertStatJ} interface in component {\tt Base} is
       implemented by component {\tt DomSetStatus1} if vertex {\tt J}
       is in $D$ and by component {\tt DomSetStatus0}
       otherwise. Observe that for each $(i_j, o_j) \in R$, this
       software system produces output $o_j$ given input $i_j$.
\item Conversely, suppose that the constructed instance of CSCreate$_{poly}$
       has a working component-based software system based on $c$
       relative to $L_{int}$, $L_{comp}$, and $R$. In order to
       accommodate all input-output pairs in $R$, the $|V|$ 
       {\tt vertStatJ} interfaces in component {\tt Base} must be
       implemented by {\tt domSetStatus} components
       such that the produced values in array 
       {\tt vS} specify a dominating set of size $k$ in $G$. Hence,
       the existence of a working component-based software system
       for the constructed instance of CSCreate$_{poly}$ implies the existence
       of a dominating set of size $k$ for the given instance of
       {\sc Dominating set}.
\end{itemize}

\noindent
The reduction above is thus correct.
Given the $NP$-hardness of {\sc Dominating set}, this reduction 
implies that CSCreate$_{poly}$ is $NP$-hard when
$|L_{comp}| = 3$, $I_{ci} = 2$, and $S_{depth} = 2$ and
hence by Lemma \ref{LemPrmProp3} not fp-tractable relative to these parameters unless 
$P = NP$. The fp-intractability result for ESCreate$_{poly}$ then follows by 
contradiction from Observation \ref{ObsESCreateSolvability}.
\end{proof}

\vspace*{0.15in}

Finally, let us prove Result B.22. This result
addresses the feasibility of preprocessing a fixed portion of the input to an 
intractable problem to create polynomial-size information that can be used to solve 
future instances of that problem relative to a varying portion of the problem input 
in polynomial time. This will be evaluated using the computational complexity 
framework described in \cite{CD+02}. The input to a problem of interest will be 
considered as a pair $\la x, y\ra$, where $x$ is the fixed part and $y$ is the 
varying part. As we will focus here on decision problems like CSCreate$_{poly}$ in 
the main text whose solution is either ``Yes'' or ``No'', each such problem will 
be considered a language of pairs. An example of such a problem is 

\begin{center}
{\sc C-SAT} $= \{\la (f), (p)\ra ~|~ p$ can be extended to a truth assignment satisfying $f\}$
\end{center}

\noindent
where fixed part $f$ is a formula in conjunctive normal form (a set $C = \{C_1, C_2,
\ldots C_n\}$ of clauses AND-ed together where each clause is a set of negated or
unnegated Boolean variables, e.g., $(v_1 OR \neg v_2 OR v_4) AND (v_2 OR v_3)$)
over a set of Boolean variables $V = \{v_1, v_2, \ldots, v_m\}$ and varying part $p$
is a partial truth assignment to the variables in $V$. Consider the following
reducibility between such problems.

\begin{definition}
\cite[Definition 2.8]{CD+02} A $||{\leadsto}$ reduction between two languages of
pairs $A$ and $B$ is a triple$f_1, f_2, g)$ where $f_1$ and $f_2$ are poly-size 
unary functions and $g$ is a binary polynomial-time function such that for any pair 
of strings $\la x, y\ra$ it holds that:

\begin{center}
$\la x, y\ra \in A$ if and only if $\la f_1(x, |y|), g(f_2(x, |y|), y)\ra \in B$
\end{center}
\end{definition}

\noindent
It is known that {\sc C-SAT} is hard under $||{\leadsto}$ reductions relative to a 
class ${||{\leadsto}}NP$ \cite[Section 3.1]{CD+02} and hence not preprocessable in the sense 
above unless the Polynomial Hierarchy $PH$ collapses, i.e., $PH = \Sigma^r_2$ \cite[Theorem 2.12]{CD+02}.\footnote{
Aficionados of \cite{CD+02} will protest that by Theorem 2.12, our result implies that
the Polynomial Hierarchy only collapses to $\Sigma^p_3$. However, when that theorem is
combined with the Karp-Lipton Theorem \cite{KL82} (the most typical version of
which states that if $NP \subseteq P/Poly$ then $PH = \Sigma^p_2$), we can strengthen
Theorem 2.12 to give our stated result.
}
We will now prove the result by a reduction based on that in the proof of
Result A.9 (see Figure \ref{FigRed8}).

\begin{figure}[t]
\begin{center}
\includegraphics[width=4.2in]{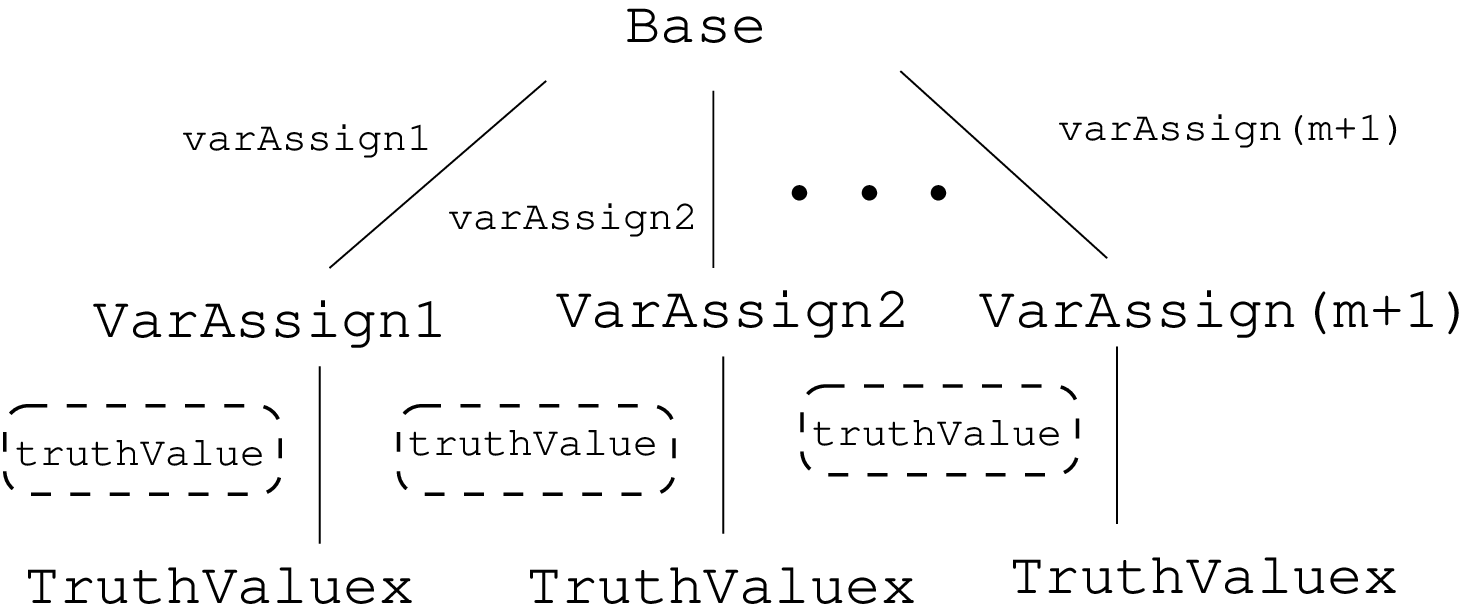}
\end{center}
\caption{General structure of valid software systems created
          by the reduction in the proof of Result B.22. Note that index $x$ in 
          {\tt TruthValuex} is such that $x \in \{True,False\}$. Following the 
          convention in Figure \ref{FigExESS3}, interfaces with multiple implementing
          components are enclosed in dashed boxes.
}
\label{FigRed8}
\end{figure}

\begin{description}
\item[{\bf Result B.22}:] If ESAdapt$_{poly}$ can have $L_{int}$, $L_{comp}$, $R$, and $c$ 
                    preprocessed to create polynomial-size information that can be 
                    used to solve ESAdapt$_{poly}$ instances of arbitrary $S$, 
                    $Rew()$, and $k$ in polynomial time then the Polynomial Hierarchy
                    $PH$ collapses, i.e., $PH = \Sigma^p_2$.
\end{description}
\begin{proof}
Consider the following $||{\leadsto}$ reduction from {\sc C-SAT} to CSAdapt$_{poly}$:
given an instance $I = \la (f), (p)\ra$ of {\sc C-SAT}, construct an instance 
$I' = \la (R, L_{int}, L_{comp}, c), \linebreak (S, Rew(), k)\ra$ of CSAdapt$_{poly}$ in which
$R$ has a single input-output pair $\{True, 1\}$. Let $f'$ be the version of $f$
in which there is a new Boolean variable $x_{m+1}$ that is OR-ed into every
clause and $p'$ be the version of $p$ such that $x_{m+1} = False$. Let 
$L_{int}$ consist of $m + 3$ interfaces broken into three groups:

\begin{enumerate}
\item{
A single interface of the form

\begin{center}
\begin{verbatim}
    interface Base {
         void main(Input I)
    }
\end{verbatim}
\end{center}
}
\item{
A set of $m + 1$ interfaces of the form

\begin{center}
\begin{verbatim}
    interface varAssignJ {
         Boolean varAssignJ()
    }
\end{verbatim}
\end{center}

\noindent
for $1 \leq {\tt J} \leq m + 1$.
}
\item{
An interface of the form

\begin{center}
\begin{verbatim}
    interface truthValue {
         Boolean truthValue()
    }
\end{verbatim}
\end{center}
}
\end{enumerate}

\noindent
Let $L_{comp}$ consist of $m + 4$ components broken into three groups:

\begin{enumerate}
\item{
A single component of the form

\begin{center}
\begin{verbatim}
    component Base provides base 
                   requires varAssign1, varAssign2, ..., 
                            varAssignm {
         void main(Input I) {

             create Boolean array vA of length m

             vA[1] = varAssign1()
             vA[2] = varAssign2()
                 ...
             vA[m+1] = varAssign(m+1)()

             if truth assignment vA satisfies f' then
                 output 1
             else
                 output 0
         }
    }
\end{verbatim}
\end{center}
}
\item{
A set of $m + 1$ components of the form

\begin{center}
\begin{verbatim}
    component VarAssignJ provides varAssignJ {
                         requires truthValue
         Boolean varAssignJ() {
             return truthValue()
         }
    }
\end{verbatim}
\end{center}

\noindent
for $1 \leq {\tt J} \leq m + 1$.
}
\item{
Two components of the form

\begin{center}
\begin{verbatim}
    component TruthValueFalse provides truthValue {
         Boolean truthValue() {
             return False
         }
    }

    component TruthValueTrue provides truthValue {
         Boolean truthValue() {
             return True
         }
    }
\end{verbatim}
\end{center}
}
\end{enumerate}

\noindent
Observe that in component {\tt Base}, if entry $i$ of array {\tt vA} 
has value $True$ ($False$), then variable $i$ in the candidate truth-assignment for 
$f$ encoded in {\tt vA} has value $True$ ($False$). Let $c$ be component 
{\tt Base} in $L_{comp}$ and $S$ be the component-based software system consisting
of {\tt Base} and {\tt VarAssign1}, {\tt VarAssign2},
\ldots, {\tt VarAssign(m+1)} such that {\tt varAssign(m+1)()} is implemented by 
{\tt TruthValueTrue} and all other \linebreak {\tt varAssignx()}, $1 \leq {\tt x} \leq m$, have
a random implementation relative to \linebreak {\tt TruthValueFalse} and {\tt TruthValueTrue}.
Observe that courtesy of the truth-setting of $v_{m+1}$ encoded in $S$, each clause
in $f'$ is satisfied and $S$ will output 1 on any input. Finally, let $Rew(S'')$ 
return the number of variable-assignments in the truth-assignment for
$f'$ encoded in component-based system $S''$ that differ from those in the partial
truth-assignment $p'$ and set $k = 0$. Note that the instance of CSAdapt$_{poly}$ 
described above can be constructed in time polynomial in the size of the given 
instance of {\sc C-SAT} and that this construction encodes the required functions
$f_1()$, $f_2()$, and $g()$ in the definition of $\leq_{comp}$-reducibility; 
moreover, any candidate software system created relative $L_{int}$, $L_{comp}$, and 
$c$ runs in time linear in the size of input $I'$.

Let us now verify the correctness of this reduction:

\begin{itemize}
\item Suppose there is a satisfying assignment $p''$ that extends $p$ and
       satisfies $f$. Construct a component-based
       software system $S'$ consisting of $c =$ {\tt Base} and
       {\tt VarAssign1}, {\tt VarAssign2}, \ldots, {\tt VarAssign(m+1)} such
       that {\tt varAssignx()}, $1 \leq {\tt x} \leq m$, implements the
       {\tt TruthVValue} component corresponding to the value of $v_x$ in $p''$
       and {\tt varAssign(m+1)()} implements {\tt TruthValueFalse}. Observe that 
       for the only requirement $r \in R$, this software system produces output $1$ 
       given input $True$; moreover, $Rew(S') = 0 \leq k = 0$.  Hence, the existence
       of a satisfying truth-assignment that extends $p$ and satisfies $f$ in
       the given instance of {\tt C-SAT} implies the existence of a working 
       component-based software system $S'$ for the constructed instance of 
       CSAdapt$_{poly}$ such that $Rew(S') \leq k$ 
\item Conversely, suppose that the constructed instance of CSAdapt$_{poly}$
       has a working component-based software system $S'$ based on $c$
       relative to $L_{int}$, $L_{comp}$, and $R$ such that $Rew(S') \leq k$. As
       $k = 0$, by the definition of $Rew()$, this means that the 
       truth-assignment to the variables in $V$ encoded in $S'$ cannot differ
       from any of the variable assignments in the partial truth-assignment $p'$.
       As $v_{m+1} = False$ in $p'$, this means that the truth-assignments to
       the remaining variables in $V$ (i) are an extension of $p$ and (ii)
       encode a truth assignment that satisfies $f$. Hence, the existence of a 
       working component-based software system for the constructed instance of 
       CSAdapt$_{poly}$ such that $Rew(S') \leq k$ implies the existence
       of a satisfying truth-assignment that extends $p$ and satisfies $f$ in
       the given instance of {\tt C-SAT}.
\end{itemize}

\noindent
The reduction above is thus correct. Given the ${||{\leadsto}}NP$-hardness of 
{\sc C-SAT}, this reduction implies that CSAdapt$_{poly}$ is ${||{\leadsto}}NP$-hard
as well and not preprocessable in the sense of \cite{CD+02} unless the
Polynomial Hierarchy $PH$ collapses, i.e., $PH = \Sigma^p_2$ (see the discussion
immediately prior to this proof). The same
non-preprocessability result for ESAdapt$_{poly}$
then follows by contradiction from Observation \ref{ObsESCreateSolvability}.
\end{proof}

\end{document}